\documentclass{IEEEtran}
\usepackage{amsmath,amsfonts,amssymb,amsthm}
\usepackage{cite}
\usepackage{float,graphicx}
\usepackage{algorithm,algorithmic}
\usepackage{bbm}
\usepackage{balance}
\usepackage{xcolor}
\usepackage{pgfplots}
\usepackage{ifthen}
\usetikzlibrary{arrows}

\newcommand*{\horzbar}{\rule[.5ex]{2.5ex}{0.5pt}}

\newtheorem{theorem}{Theorem}

\newtheorem{innercustomthm}{Theorem}
\newenvironment{customthm}[1]
  {\renewcommand\theinnercustomthm{#1}\innercustomthm}
  {\endinnercustomthm}
  
\newtheorem{innercustomlemma}{Lemma}
\newenvironment{customlemma}[1]
  {\renewcommand\theinnercustomlemma{#1}\innercustomlemma}
  {\endinnercustomlemma}
  
\newtheorem{innercustomprop}{Proposition}
\newenvironment{customprop}[1]
  {\renewcommand\theinnercustomprop{#1}\innercustomprop}
  {\endinnercustomprop}  

\newtheorem{lemma}{Lemma}
\newtheorem{claim}{Claim}
\newtheorem{defn}{Definition}
\newtheorem{conj}{Conjecture}
\newtheorem{coro}{Corollary}
\newtheorem{prop}{Proposition}

\begin{document}

\title{Blind Joint MIMO Channel \\Estimation and Decoding}
\author{Thomas~R.~Dean,~\IEEEmembership{Member,~IEEE,}
Mary~Wootters,~\IEEEmembership{Member,~IEEE,} and 
Andrea~J.~Goldsmith,~\IEEEmembership{Fellow,~IEEE}%
\thanks{T. Dean is with the Department of Electrical Engineering, Stanford, CA
94305 USA (e-mail: trdean@stanford.edu).}%
\thanks{M. Wootters is with the Department of Computer Science and the Department of
Electrical Engineering, Stanford, CA 94305 USA (e-mail: marykw@stanford.edu).}%
\thanks{A. Goldsmith is with the Department of Electrical Engineering, Stanford, CA
94305 USA (email: andrea@ee.stanford.edu).}
\thanks{This work was presented in part in 2017 at IEEE Globecom in Singapore. T. Dean is supported by the Fannie and John Hertz Foundation. This work was supported in part by the NSF Center for Science of Information under Grant CCF-0939370}}
\maketitle

\global\csname @topnum\endcsname 0
\global\csname @botnum\endcsname 0

\begin{abstract}
We propose a method for MIMO decoding when channel state information (CSI) is unknown to both the
transmitter and receiver.  The proposed method requires some structure in the transmitted signal for the decoding to be effective, in particular that the underlying 
sources are drawn from a hypercubic space.   Our
proposed technique fits a minimum volume parallelepiped to the received samples.
This problem can be expressed as a non-convex optimization problem that
can be solved with high probability by gradient descent. 
Our blind decoding algorithm can be used when communicating over unknown MIMO wireless 
channels using either BPSK or MPAM modulation. We apply our technique to jointly estimate 
MIMO channel gain matrices and decode the underlying transmissions with only
knowledge of the transmitted constellation and without the use of pilot symbols.
Our results provide theoretical guarantees that the proposed
algorithm is correct when applied to small MIMO systems.
Empirical results show small sample size requirements, making this
algorithm suitable for block-fading channels with coherence times typically seen in practice.
Our approach has a loss of less than 3dB compared to zero-forcing with
perfect CSI, imposing a similar performance penalty as space-time coding techniques without the loss of rate incurred by those techniques.

\end{abstract}

\section{Introduction}
In this work we propose a method to blindly estimate MIMO channels
and decode the underlying transmissions. Given only knowledge of the statistics of
the channel gain matrix, the constellation, and the channel noise, we exploit 
the geometry of the constellation in order to jointly estimate the channel gain 
matrix and decode the underlying data.  More precisely, we exploit the fact that
the underlying constellation is often \emph{hypercubic}, i.e. forms a regular
$n$-dimensional polytope with mutually perpendicular sides, as is the case with BPSK or
MPAM modulation.  The technique presented in this work can also 
be applied to decoding and estimation in the SIMO MAC, where channel gains are unknown at the receiver and there is no coordination among transmitters.

In modern cellular systems, there is
up to 15\% transmission overhead dedicated to performing channel estimation 
\cite{3gpp}.  Improving channel estimation techniques, through, for example, 
sparse dictionary learning \cite{prasad2014joint}, is an active area of
research.  In practice, channel state information (CSI) is not always needed to
decode, but 
schemes that communicate without CSI impose 
losses in rate or increased symbol error rates \cite{tarokh1999space}.  Additionally, blind decoding schemes for MIMO systems exist, but they are often inefficient in terms of complexity or sample size requirements.  This is discussed in more detail in Section \ref{sec:related}.  Obtaining
accurate channel estimates is likely to become more challenging in future generation wireless
systems, which will likely have both increased spatial diversity and decreased coherence times 
\cite{chowdhury2016scaling}. Hence, improvements to channel estimation, or to schemes 
that communicate without CSI, have the potential to reduce overhead as well as
improve performance in current and future wireless systems.

This work is also motived by research in physical-layer security.  Several
works have proposed keyless authentication schemes that attempt to identify
users based on properties of the physical channel over which they communicate
(see, for example \cite{jorswieck2015broadcasting}, or \cite{tomasin2017survey}
for a survey).
Since MIMO channels are often well conditioned, and hence invertible, such
schemes require that CSI remains hidden from an adversary.  Our work shows that MIMO systems inherently leak CSI when the underlying source is structured.  
From a security perspective, this work implies that an eavesdropper need not have 
knowledge of pilot symbols nor any knowledge of the data being transmitted in order 
to efficiently intercept and decode MIMO communications.  Any
scheme that attempts to provide security by hiding or obscuring pilot symbols
will be insecure.

The blind decoding technique introduced in our work
is motivated by a classical problem in convex optimization: fitting a minimum 
volume ellipsoid (also known as the L\"{o}wner-John ellipsoid) to a set of samples, as given in
\cite{gls88}.  The method proposed in this work fits samples to 
within a parallelepiped, i.e. an $n$-dimensional polytope that has 
parallel and congruent opposite faces, thereby recovering the inverse of the channel gain matrix. In this work, we focus on MIMO systems that have a small number of
transmit antennas, specifically up to 8; this choice of parameters captures
nearly all MIMO systems in use in wireless systems deployed today, for example see \cite{3gppRI} or \cite{80211RI}.

We outline the major contributions of this work as follows:
\begin{itemize}
\item We introduce a novel (non-convex) optimization problem, whose solutions
capture those of the blind decoding problem.  Our formulation exploits the 
structure of the underlying constellation so that solving this optimization problem both
estimates the channel gain matrix and detects the underlying data
symbols using far fewer samples of received symbols than previous blind decoding 
techniques.
\item Despite the fact that this problem is non-convex, we give both theoretical and empirical results showing that gradient descent is effective for solving the blind MIMO decoding problem.  More precisely, for general values of $n$, we derive sufficient conditions to ensure that global optima correspond to solutions of the blind decoding problem for the case where BPSK is transmitted and in the limit of infinite SNR.  We further relate the blind decoding problem to the Hadamard Maximal Determinant problem. For $n \leq 4$, we present necessary conditions so that there are no spurious optima within the domain of the optimization problem, implying that our approach will always return a solution to the blind decoding problem.  Further, we provide evidence that 
formulating equivalent necessary conditions for larger values of $n$ is likely intractable.  
\item For $n \leq 4$, we give theoretical results that relate the number of observed samples to the probability that our method returns a correct solution to the blind decoding problem.  Our theoretical results nearly exactly match our empirical results.  Notice that $n \leq 4$ captures the majority of MIMO systems in use today.
\item Although it seems difficult to provide theoretical evidence that gradient descent performs well for large values of $n$, we present empirical evidence suggesting that gradient descent efficiently solves
 our non-convex optimization problem and the blind decoding problem for values of $n$ as high as $n=15$ and for values of $M$ as large as $M=16$.  
Further, our
empirical evidence shows that our approach is robust in the presence of AWGN and
that decoding performance is comparable to known methods that communicate over a
MIMO channel without CSI or with imperfect CSI. In particular, our blind method
\emph{outperforms} existing non-blind methods when the CSI is somewhat inaccurate. 
\end{itemize}

The remainder of the paper is organized as follows.  In Section \ref{sec:related} we provide a survey of techniques related to our work.  Section \ref{sec:model} describes our system model.  Section \ref{sec:fit} outlines the optimization problem that solves the blind decoding problem, as well as algorithms that solve this optimization problem; the theoretical performance of these algorithms is shown in Section \ref{sec:prelim}.  Section \ref{sec:empirical} presents empirical results that support the theory contained in Section \ref{sec:prelim}.  Concluding remarks are provided in Section \ref{sec:conc}. Proofs not contained in Section \ref{sec:prelim} are found in the appendices.  

\section{Related Work}
\label{sec:related}
The problem of joint blind channel estimation and decoding is not new.  For
example, in
\cite{talwar1996blind}, the authors apply MMSE techniques to the blind decoding
of MIMO problems over small alphabets while simultaneously recovering the
underlying channel gain matrix. The approach in \cite{talwar1996blind} requires
the number of samples of received signals used by the algorithm to grow linearly 
with constellation size, which is
exponential in $n$, the number of transmit antennas.  The approach in \cite{talwar1996blind} 
only requires the underlying constellation to be discrete; however, for
constellations that are also hypercubic, our approach requires far 
fewer received samples than the approach of \cite{talwar1996blind} based on our simulation results.

In addition, blind decoding algorithms have previously been applied to hypercubic 
sources.  In \cite{hansen1997hyperplane},  the authors present a statistical
learning algorithm that applies a modified version of the Gram-Schmidt algorithm
to an estimate of the covariance matrix of the received signals to learn the
channel gain matrix.  In a different setting, the authors in
\cite{nguyen2006learning} learn a parallelepiped from a covariance matrix by
first orthogonalizing and then recovering the rotation through higher order
statistics.  
Our method does not rely on the covariance matrix estimation and our empirical results show that it requires fewer
samples than the techniques of \cite{hansen1997hyperplane} and
\cite{nguyen2006learning}, especially when the channel gain matrix has a high condition number.

Blind source separation is the separation of a set of unknown signals that are mixed 
through an unknown (typically linear) process with
no or little information about the mixing process or the source signals.
Several previous works have considered using blind source separation techniques
for detection of signals transmitted over unknown MIMO channels. 
Blind source separation is typically accomplished through techniques such as 
Principle Component Analysis (PCA), Independent Component Analysis (ICA), or 
Non-Negative Matrix Factorization (NMF); see \cite{hyvarinen2000independent} 
for a survey of these techniques.
Other techniques exploit structure in the mixing process, for example, 
\cite{ling2015blind} requires the mixing process to be a Toeplitz matrix.
 Our
technique differs from traditional blind source separation as we obtain an
estimate of the source signals by learning the inverse of the mixing process rather than directly estimating the source signals.
As an output, our algorithm produces both an estimate of the mixing process,
i.e. the channel gain matrix, and the source signal, i.e. the transmitted symbols.
Similarly, blind channel estimation techniques have been studied, although most
commonly for SISO channels.  See \cite{johnson1998blind} or
\cite{benveniste1984blind} for surveys on this topic.
The approach presented in this paper can be viewed outside the
context of communicating over an unknown MIMO channel as a general technique
that performs blind source separation of sources mixed through an unknown, linear process.

Many techniques exist for communications over unknown MIMO
channels that do not rely on channel estimation.  
For example, Space-Time Block Coding (STBC) was 
introduced by Alamouti
in \cite{alamouti1998simple} and formalized by Tarokh \emph{et al}. in
\cite{tarokh1999space}.  
These techniques rely on coding transmissions
using sets of highly orthogonal codes so that the receiver can recover the
transmission without CSI. For the case of two transmitters, rate one
space-time block codes exist that impose a 3 dB penalty in terms of SNR at the
receiver.  For larger numbers of transmitters, rate one codes do not exist.  
Our techniques do not require any coding at the transmitter and thus do not impose
any rate penalty.  Numerical simulation shows the decoding performance of our
technique to be comparable to rate one STBC methods. 

\section{System Model and Notation}
\label{sec:model}
This work focuses on an $n \times n$ real-valued MIMO channel with block fading and AWGN. In Section \ref{sub:extentions}, we discuss how this work can be extended to complex-valued channels and channels with more receivers than transmitters.  The input-output relation of this channel is characterized by:
\begin{equation}
\label{eq:mimo}
\mathbf{y} = \mathbf{A} \mathbf{x} + \mathbf{e},
\end{equation}
where $\mathbf{x}$ is drawn from a standard $M$-PAM or BPSK constellation; that
is, $\mathcal{X}^n$ for $\mathcal{X} = \{2i-1-M:i=1,2,\ldots,M\}$ or $\{-1, +1\}$
respectively.  The channel matrix $\mathbf{A} \in \mathbb{R}^{n \times n}$ is drawn from a random distribution.  For the simulations in Section \ref{sec:empirical}, we take $\mathbf{A}$ to be drawn with i.i.d. entries from $\mathcal{N}(0,1)$; however, our approach only requires $\mathbf{A}$ to be full rank and thus $\mathbf{A}$ can be considered to be drawn from an arbitrary distribution or entirely deterministic.  The noise
$\mathbf{e} \in \mathbb{R}^{n}$ has i.i.d. entries drawn from $\mathcal{N}(0,\sigma^2)$.  
We assume that
$\mathbf{A}$ is block-fading, meaning that the value of $\mathbf{A}$ remains
constant for some coherence time, $T_c$, after which $\mathbf{A}$ is redrawn. 

The receiver sees samples $\mathbf{y}_1, \hdots, \mathbf{y}_k$, as in
(\ref{eq:mimo}).  We
assume that the receiver knows the constellation $\mathcal{X}^n$ but has no knowledge of
the points drawn from it, nor does it have any knowledge of the matrix $\mathbf{A}$.

Given messages $\mathbf{x}_1, \ldots, \mathbf{x}_k \in \mathbb{R}^n$, we
denote by $\mathbf{X}$ the $n \times k$-dimensional matrix formed by taking each
symbol as a column, and by $\mathbf{Y}$ the corresponding matrix with received
symbols as columns.  Notice that we cannot hope to recover $\mathbf{A},\mathbf{X}$ 
exactly.  Indeed, since 
the constellation is invariant under sign flips and permutations, we can always 
write $\mathbf{AX} = \mathbf{ATT}^{-1}\mathbf{X}$, where $\mathbf{T}$ is the 
product of a permutation matrix and a diagonal matrix with entries $\pm 1$, and 
there is no way to distinguish between the solutions $(\mathbf{A}, \mathbf{X})$ and
$(\mathbf{AT},\mathbf{T}^{-1}\mathbf{X})$.   Such a matrix $\mathbf{T}$ is
termed an admissible transform matrix (ATM) in \cite{talwar1996blind}. Thus, in
this work, we aim to recover $\mathbf{AT}$ for some ATM $\mathbf{T}$. 

While inevitable, these sign and permutation ambiguities do not pose a huge problem in
practice, and we ignore them when comparing the
results to MIMO decoding algorithms with known CSI.  We justify this approach as
follows.   First, in the non-blind estimation case (i.e. where we have some
control over the transmission scheme and allow the transmitter to send pilot
symbols), assuming $M>n$, the permutation ambiguity could be resolved by a single pilot symbol. Additionally, if we consider the SIMO Multiple Access Channel, we can ignore the issue of 
permutations of the received signals, for example by assuming that identification
occurs at a higher protocol layer.  Finally, we note that the sign ambiguity
can easily be resolved through differential modulation.
In practice, it may also be possible to resolve these ambiguities by
examining structure in the transmission scheme, present from either
protocol/framing data or structure in the underlying data.  
This could prove to
be difficult, however, if the data is encrypted or compressed, or the underlying
transmission protocol is designed to thwart such analysis.

The notation $\lfloor \mathbf{A} \rceil$ rounds elements of $\mathbf{A}$ to the 
nearest element of $\mathcal{X}$, and $\kappa( \mathbf{A} )$ denotes the condition 
number of the matrix $\mathbf{A}$, which is the ratio of the largest to the
smallest singular value of $\mathbf{A}$.  $\mathbf{a}_i$ denotes a column vector
formed from the $i$th column of $\mathbf{A}$ and $\mathbf{a}^{(j)}$ denotes a
row vector formed from the $j$th row of $\mathbf{A}$.
We define ${n \brack k}_q$ to be the Gaussian binomial coefficient, which for
any prime power $q$, counts the number of $k$ dimensional subspaces in a vector
space of dimension $n$ over a finite field with $q$ elements.
For two vector spaces $\mathbf{X}$ and $\mathbf{Y}$, the notation $\mathbf{X}
\leq \mathbf{Y}$ denotes ``$\mathbf{X}$ is a subspace of $\mathbf{Y}$''.  For any
$\mathbf{X} \in \mathbb{R}^{n \times n}$, $\text{vec}(\mathbf{X})$ corresponds
to the length $n^2$ column vector obtained by stacking the columns of
$\mathbf{X}$ in the usual ordering.  Given a matrix $\mathbf{X}$, the set $\text{cols}(\mathbf{X})$ denotes the set of vectors that comprise the columns of $\mathbf{X}$.

\section{Fitting a Parallelepiped}
\label{sec:fit}
Since each transmitted symbol $\mathbf{x}_i$ is drawn from a hypercube, the
values $\mathbf{A x}_i$ are contained in an $n$-dimensional parallelepiped. As
$\mathbf{y}_i = \mathbf{Ax}_i + \mathbf{e}$, the received symbols $\mathbf{y}_i$
will lie in a slightly distorted parallelepiped. At reasonable SNR levels, this
distortion will be minimal. Thus,
we formulate the problem of blindly estimating the channel gain matrix as fitting a 
parallelepiped to our observed symbols and express this problem as
an optimization problem.  Given a set of $k$ samples of
$\mathbf{y}_1,\ldots,\mathbf{y}_k$, consider the program:
\begin{align}
& \underset{\mathbf{U}}{\text{maximize}} 
& & \log | \det \mathbf{U} | \label{eq:obj}\\
& \text{subject to}
& & \| \mathbf{Uy}_i \|_\infty \leq M + c \cdot \sigma, \; i = 1, \ldots, k.
\label{eq:const}
\end{align}
The domain of $\mathbf{U}$ is all $n \times n$ invertible matrices (not
necessarily symmetric or positive-semidefinite), meaning the
objective function is not necessarily convex. However, we will show that if some
condition on $\mathbf{X}$ is satisfied, then solutions in the form
$\mathbf{U} = \mathbf{TA}^{-1}$ for some ATM $\mathbf{T}$, correspond to global
optima to this problem; moreover, we show that these are often the only optima and
that gradient descent will find them.  

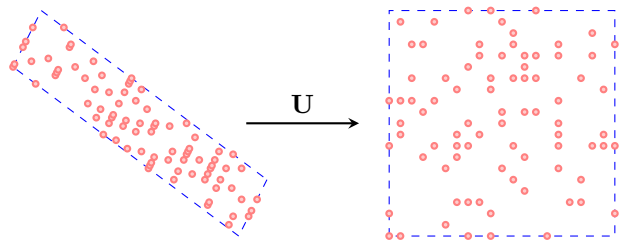
\begin{figure} 
    \begin{centering}
    
\begin{tikzpicture}
[point/.style={circle,draw=red!50,fill=red!20,thick,
inner sep=0pt,minimum size=0.75mm}]
\draw[dashed, blue] (0,2.25) -- (0.375,3);
\draw[dashed, blue] (0.375,3) -- (3.375,0.75);
\draw[dashed, blue] (3.375,0.75) -- (3,0);
\draw[dashed, blue] (3,0) -- (0,2.25);

\draw[thick, black, -stealth] (3.1,1.5) -- (4.6,1.5)
	node [above,align=center,midway] {$\mathbf{U}$};

\filldraw[fill=none, draw=blue, dashed] (5,0) rectangle (8,3);

\foreach \x in {-1,-0.9,...,1.1}
{
	\foreach \y in {-1,-0.9,...,1.1}
    {
    \pgfmathparse{rnd < 0.2 ? int(1) : int(0)}
    \ifnum\pgfmathresult=1 
    	\node at (1.5*\x+6.5, 1.5*\y+1.5) [point] {};
        \node at (1.5*0.125*\x+1.5*\y+1.6875, 1.5*0.25*\x-1.5*0.75*\y+1.5) [point] {};
    \else
       
    \fi	
    }
}
\end{tikzpicture}

    \par\end{centering}
    \caption{The program given in (\ref{eq:obj})--(\ref{eq:const}) aims to find a linear transformation $\mathbf{U}$, that transforms a set of observed MIMO samples $\mathbf{y}_i$ (left), such that the resulting $\mathbf{Uy}_i$ lie within the unit $\ell_\infty$ ball (right).  By finding a $\mathbf{U}$ that has a maximally-valued determinant, we are effectively finding a parallelepiped of minimal volume that fits the observed samples.}
    \label{fig:inverse}
\end{figure}

In this section, we present three
separate algorithms. We first present Algorithm \ref{alg:fit}, a simple algorithm 
using gradient descent in the usual manner to solve
(\ref{eq:obj})--(\ref{eq:const}). We demonstrate in Section \ref{sec:empirical} that, in practice, this algorithm is sufficient to
recover a solution to the blind decoding problem.  In Section
\ref{sub:mod}, we present Algorithm \ref{alg:fit_corner}, a slightly
modified version of Algorithm \ref{alg:fit} that has a theoretical guarantee of
correctness under conditions given in Section \ref{sec:prelim}.  Finally, in
Section \ref{sub:interior}, we include a description of the interior-point
method which allows us to use ordinary gradient descent to solve
(\ref{eq:obj}) while remaining in the feasible region as given by
(\ref{eq:const}).  

Informally, by seeking to maximize the determinant 
of $\mathbf{U}$, subject to the $\ell_\infty$-norm constraints, 
we are finding the minimum volume parallelepiped that fits the observed samples.
Since $\mathbf{U}$ is the inverse of $\mathbf{A}$, up to an ATM, maximizing
$\mathbf{U}$ is effectively finding the minimal $\mathbf{U}^{-1}$ which maps
the $\ell_\infty$-ball to the observed samples. This is depicted in Figure \ref{fig:inverse}. The quantity $c \cdot \sigma$ present in (\ref{eq:const}) adds a margin to our constraint to account for the presence of AWGN.
In practice, values close to $c=3$ appear to be optimal as this captures 99\%
of the additive Gaussian channel noise.
More careful analysis is warranted to understand how the performance of this
algorithm is affected by the value of $c$.  More optimal methods of margining
the constraints of our optimization problem, such as the method of ellipsoid
peeling, given in \cite{silverman1980minimum}, or other methods presented in \cite{boyd_opt}, may lead to further improvements in performance. However, the simple margin presented here works well in practice.

In order for (\ref{eq:obj})--(\ref{eq:const}) to be a meaningful problem, we
require $\mathbf{Y}$ to be full rank.  
If 
$\mathbf{Y}$ is not full rank, then the maximum does not exist, formally shown
in Proposition \ref{prop:fullrank}
below, which is proven in Appendix \ref{apx:prop_one}.  
\begin{prop}
\label{prop:fullrank}
If the matrix $\mathbf{Y}$ is not full rank, then (\ref{eq:obj})--(\ref{eq:const}) 
is unbounded above. Conversely, if $\mathbf{Y}$ is full rank and $k \geq n$, then
(\ref{eq:obj})--(\ref{eq:const}) is bounded above and feasible.  
\end{prop}

In our model, $\mathbf{Y}$ will be full rank with high probability and thus, we assume that $\mathbf{Y}$ is always full rank and turn our attention to
solving (\ref{eq:obj})--(\ref{eq:const}).
Maximizing the determinant of a positive-semidefinite matrix is a classic problem in convex
optimization.  Unfortunately, the matrix $\mathbf{A}$ is not
necessarily even symmetric and the problem is not convex.  In order to solve the 
problem we apply the MATLAB \texttt{fmincon} solver that uses gradient descent to 
solve non-linear conic optimization problems.  The gradient of the problem given in
(\ref{eq:obj})--(\ref{eq:const}) has the following value (see, for example, \cite{petersen2008matrix}):
\begin{equation}
\label{eq:gradient}
\nabla \left( \log | \det \mathbf{U} | \right) =
\left(\mathbf{U}^{-1} \right)^\intercal.
\end{equation}

Before we begin gradient descent, we check that $\mathbf{Y}$ is well conditioned.  
As noted
above, $\mathbf{Y}$ must be full rank for the problem to make sense; however,
if $\mathbf{Y}$ is full rank but poorly conditioned, similar issues will arise
and $\mathbf{U}$ may not invert the channel. Thus, we return
$\texttt{FAIL}$ if $\kappa(\mathbf{Y})$, the condition number of $\mathbf{Y}$, is
larger than $\kappa_{\max}$.
The gradient descent algorithm requires a starting point as input, 
denoted as $\mathbf{U}^{(0)}$. We draw this matrix uniformly at random over the
set of all orthogonal matrices, $O(n)$, using the method described in \cite{DS87}. 
We check that this random $\mathbf{U}^{(0)}$ in fact
satisfies the constraints; if it does not, we generate a new random matrix and
scale the matrix by a constant term until we find a suitable starting condition.
This is guaranteed to find a suitable $\mathbf{U}^{(0)}$ in at most $\log_2
\kappa_{\max}$ iterations. The algorithm is summarized as Algorithm
\ref{alg:fit} below. 

 \begin{algorithm}
 \caption{Fitting a Parallelepiped}
 \label{alg:fit}
 \begin{algorithmic}[1]
 \renewcommand{\algorithmicrequire}{\textbf{Input:}}
 \renewcommand{\algorithmicensure}{\textbf{Output:}}
 \REQUIRE An $n \times k$ matrix of received samples $\mathbf{Y}$.
 \ENSURE  An estimate of the inverse of the channel gain matrix $\mathbf{U}$,
and an estimate of the transmitted symbols $\hat{\mathbf{X}}$
  \IF {$\kappa(\mathbf{Y}) > \kappa_{\mathrm{max}}$}
    \RETURN \texttt{FAIL}
  \ENDIF
  \STATE scale = 1;
  \STATE Draw $\mathbf{U}^{(0)}$ uniformly from O($n$);
  \WHILE {$|\mathbf{U}^{(0)} \mathbf{y}_i|_\infty > 1$ for any $i$}
  \STATE scale = scale * 2;
  \STATE Draw $\mathbf{U}^{(0)}$ uniformly from O($n$);
  \STATE $\mathbf{U}^{(0)}$ = $\mathbf{U}^{(0)}$ / scale;
  \ENDWHILE
  \STATE Run gradient descent over (\ref{eq:obj})--(\ref{eq:const}) starting at 
	$\mathbf{U}^{(0)}$ to find an optimal value of $\mathbf{U}$;
 \RETURN $\mathbf{U}$, $\hat{\mathbf{X}} = \lfloor \mathbf{UY} \rceil$.
 \end{algorithmic} 
 \end{algorithm}

\subsection{Modified Gradient Descent}
\label{sub:mod}
In practice, Algorithm \ref{alg:fit} works well and empirical results show that
it always returns solutions to (\ref{eq:obj})--(\ref{eq:const}) when $k$ is
sufficiently large.  However, the problem geometry, which is studied in Section
\ref{sec:prelim}, is non-convex and there is in fact a small but non-zero probability that gradient
descent will not terminate at a global optimum.  Moreover, in Section \ref{sec:prelim}, we show that Algorithm \ref{alg:fit} will only fail to find a global optima at specific dimensions, and further that its probability of failure is low. 
In this subsection, we present a
modified gradient descent algorithm, Algorithm \ref{alg:fit_corner}, shown below, which is 
motivated by the theory in Section \ref{sec:prelim}, where we show that
all strict solutions to (\ref{eq:obj})--(\ref{eq:const}) lie on vertices of the problem
boundary. Algorithm \ref{alg:fit_corner} always terminates at a vertex of the feasible 
region, and it is conjectured that, for general $n$, when $k$ is slightly larger than $n$, all
non-singular vertices are global optima and solutions to our channel estimation
problem, implying that Algorithm \ref{alg:fit_corner} will always be correct.

Before describing Algorithm \ref{alg:fit_corner}, we make the following
observations. The feasible region is bounded by $2kn$ halfspaces, forming an
$\mathbb{R}^{n^2}$ dimensional polytope.  We denote this polytope as $\mathcal{P}$.
Notice that any row of $\mathbf{U}$ can be changed without effecting
whether or not the constraints on each of the other rows of $\mathbf{U}$ are 
satisfied. Further, we say that 
$\mathbf{U}$ is a vertex if it is a vertex of the $n^2$-dimensional 
polytope which defines the problem boundary. 

Algorithm \ref{alg:fit_corner} begins by choosing a starting position in the
same manner as Algorithm \ref{alg:fit} and performing gradient descent.  
In Section \ref{sec:prelim} it is shown that not only are all optima
contained on the problem boundary but also that the only possible
critical points on the problem boundary exist as low-dimensional affine
subspaces, along which the objective function is constant valued.  
If gradient descent reaches such a subspace, then Algorithm
\ref{alg:fit_corner} continues by choosing a direction on this subspace at
random and moving in that direction until the edge of the feasible region is
reached.  At this point, the algorithm has either reached a vertex, in which
case it terminates, or gradient descent is continued from this point. This
process can be repeated until a vertex is reached.

 \begin{algorithm}
 \caption{Modified Gradient Descent}
 \label{alg:fit_corner}
 \begin{algorithmic}[1]
 \renewcommand{\algorithmicrequire}{\textbf{Input:}}
 \renewcommand{\algorithmicensure}{\textbf{Output:}}
 \REQUIRE An $n \times k$ matrix of received samples $\mathbf{Y}$.
 \ENSURE  An estimate of the inverse of the channel gain matrix $\mathbf{U}$,
and an estimate of the transmitted symbols $\hat{\mathbf{X}}$
  \IF {$\kappa(\mathbf{Y}) > \kappa_{\mathrm{max}}$}
    \RETURN \texttt{FAIL}
  \ENDIF
  \STATE Generate $\mathbf{U}^{(0)}$ as described in Algorithm 1.
  \WHILE {$\mathbf{U}^{(i)}$ is not at a vertex of the feasible region}
  \STATE Run gradient descent over (\ref{eq:obj})--(\ref{eq:const}) starting at
	$\mathbf{U}^{(i)}$
  \FOR {$j$ in $\{0,\ldots, n-1\}$}
 	\IF {$\mathbf{U}^{(i)}_j$ is not at a vertex}
		\STATE Move in the direction of zero gradient to a vertex;
	\ENDIF
 \ENDFOR
 \STATE $\mathbf{U}^{(i+1)} = \mathbf{U}^{(i)}$; $i$++;
 \ENDWHILE
 \RETURN $\mathbf{U}$, $\hat{\mathbf{X}} = \lfloor \mathbf{UY} \rceil$.
 \end{algorithmic} 
 \end{algorithm}

\subsection{Interior-Point Method}
\label{sub:interior}
Both Algorithms \ref{alg:fit} and \ref{alg:fit_corner} perform gradient descent
on an objective function subject to a convex set of constraints.  A na\"ive
implementation of gradient descent will not stay within these constraints.
There are many algorithms to perform constrained optimization, for an overview,
see \cite{boyd_opt}.  

For completeness, we propose using an interior-point method with a standard logarithm barrier
function to perform the gradient descent step in both Algorithms \ref{alg:fit}
and \ref{alg:fit_corner}.
\footnote{
In Section \ref{sec:prelim}, we prove that optima of our problem lie on the
boundary of the feasible region.  One may notice that the
interior-point method is not the most efficient algorithm given this fact.  We
base the analysis contained in this paper on gradient descent because it makes
the analysis tractable and more easily understood.  We defer to investigating more
efficient algorithms for this problem to be a topic of future research.
}
This method is attractive because it is simple to
implement and has reasonable computational complexity and numerical stability.
The results in Section \ref{sec:empirical} are obtained using this method.  We
formulate (\ref{eq:obj})--(\ref{eq:const}) into an unconstrained optimization 
function by using the following barrier function:
\begin{equation}
B(\mathbf{U}, \mu) = f(\mathbf{U}) - \mu \sum_{i=1}^k \sum_{j=1}^n \left(
\log ( \mathbbm{1} - u_j y_i ) + 
\log ( \mathbbm{1} + u_j y_i ) \right)
\end{equation}
The gradient of the barrier function can be computed using the expression
derived in (\ref{eq:gradient}).  This is given by:
\begin{align}
\nabla B (\mathbf{U}, \mu) = \left(\mathbf{U}^{-1}\right)^\intercal 
&- \mu \sum_{i=1}^k \sum_{j=1}^n 
\frac{1}{1-\mathbf{u}_j \mathbf{y}_i} \mathbf{e}_j \mathbf{y}_i^\intercal
\nonumber \\
&+ \mu \sum_{i=1}^k \sum_{j=1}^n 
\frac{1}{1+\mathbf{u}_j \mathbf{y}_i} \mathbf{e}_j \mathbf{y}_i^\intercal
\end{align}
where $\mathbf{e}_j$ is the $j$-th standard basis of $\mathbb{R}^n$. Notice that
this will take $O\left(n^2 k\right)$ operations per step.  The dominating
operation at each step is computing the product $\mathbf{UY}$.

\subsection{Further Extensions}
\label{sub:extentions}
The results in this paper readily extend to complex channels.  We can accomplish this by mapping an $n \times n$-dimensional complex channel to a $2n \times 2n$-dimensional real channel in the usual manner.  Note that this mapping imposes additional constraints on our optimization problem.  However, in Section \ref{sec:prelim}, we derive the necessary and sufficient conditions for our Algorithm 2 to return a correct solution to the blind decoding problem.  These results directly imply that we may simply ignore these constraints and solve the $2n \times 2n$-dimensional real problem by using Algorithm 2 on (\ref{eq:obj})--(\ref{eq:const}).  Since this approach will return the correct channel gain matrix, up to a factor of a $2n$-dimensional ATM, the amount of side information needed to recover this ATM will be identical to the $2n$-dimensional real case.  Whether or not the structure present in complex channels can be utilized to create a more efficient algorithm or reduce the required amount of side information is an open question.

When there
are more receivers than transmitters, the receiver may still apply our algorithms
by simply discarding all but $n$ received signals, but this is clearly
suboptimal.  Further optimization of this case is a topic of future research.  When there are more transmitters than receivers then the nullspace of the channel gain
matrix will always be non-trivial and thus (\ref{eq:obj})--(\ref{eq:const}) will
be unbounded above.  In this case, if we assume that the transmit signals are uncoded and the channel gain matrix is full rank, as is the case in this work, then detecting signals transmitted over this channel is not a meaningful problem.


\section{Theoretical Performance Guarantees}
\label{sec:prelim}
Proving correctness of an algorithm that solves a non-convex problem is often a difficult task.
In this section, we lay the groundwork for such an analysis by studying the noiseless case.  We
provide guarantees on the correctness of Algorithm \ref{alg:fit_corner} for
$n=2,3$, and $4$. The motivation for studying Algorithm \ref{alg:fit_corner}
over Algorithm \ref{alg:fit} will become apparent in the following subsections,
as will the difficultly of proving the correctness of our algorithms for more general or
larger values of $n$.  The results in this section are strongly supported by the empirical results shown in Section \ref{sec:empirical}.

For the
results in this section, we suppose $\sigma = 0$.  We also focus on the BPSK
case, so $\mathbf{x}_i \in \{-1, +1\}^n$. Deriving matching theoretical results for larger
$M$ and in the presence of noise remains an open problem; however, empirical results, given in Section \ref{sec:empirical}, show that Algorithms 1 and 2 still work extremely well in these cases.
As mentioned in Section \ref{sec:fit}
the problem (\ref{eq:obj})--(\ref{eq:const}) is a non-convex optimization
problem, 
and thus has several optima.  
Our analysis of gradient descent applied to this problem will proceed
as follows. First,
we will show in Section \ref{sub:had_red} that if $n=k$, then the optimization problem reduces to the
\emph{Hadamard Maximal Determinant problem}, which asks for the maximum value of
an $n$-dimensional matrix whose entries are contained on the unit disk.  
We will use this result to establish
guiding intuition for the remainder of this section.  Additionally, we show that completely understanding the problem geometry when $n=k$ would solve the Hadamard Maximal Determinant Problem; since the latter is considered extremely difficult, this implies that a complete theoretical analysis of our problem is likely out of reach.

In this section, we present a set of theorems that describe when and why Algorithms \ref{alg:fit} and \ref{alg:fit_corner} correctly solve the blind decoding problem.  The proofs of these theorems are contained within the appendices of this work.  The remainder of this section is organized as follows.
In Section \ref{sub:behavior} we will show 
that Algorithm \ref{alg:fit_corner} will always terminate at a vertex of the feasible 
region and that all strict optima of (\ref{eq:obj}) lie on these vertices.
In Section \ref{sub:vertex}, we will state the necessary conditions under which the 
set of global optima contains the solutions to the blind decoding problem. 
Finally, we will conclude by stating our theoretical guarantees; namely, necessary and sufficient conditions for Algorithm \ref{alg:fit_corner} to correctly solve the blind decoding problem for the cases $n=2,3,$ and
$4$. Additionally, we conjecture about the performance of Algorithms \ref{alg:fit} and
\ref{alg:fit_corner} for larger $n$.  Note that in practice, values of $n \leq 4$ captures nearly all MIMO systems that are currently in use today.

\subsection{Reduction to the Hadamard Maximal Determinant Problem}
\label{sub:had_red}
We now proceed by showing the equivalence between (\ref{eq:obj})--(\ref{eq:const}) and
the Hadamard Maximal Determinant problem for the case $n=k$.  This problem is related to finding dimensions at which Hadamard matrices exist.  A Hadamard matrix is a $\{-1,+1\}$-valued matrix with mutually 
orthogonal rows and columns. Hadamard matrices are known to exist for $n=1,2^k$, for 
all $k \in \mathbb{N}$, and are conjectured to exist when $n \equiv 0 \mod 4$.
\begin{lemma}
\label{lem:had_red}
There exists an efficient algorithm that solves
(\ref{eq:obj})--(\ref{eq:const}) when $n=k$, if and only if there exists an
efficient solution to the Hadamard Maximal Determinant problem.
\end{lemma}
\begin{proof}
We show how, given an efficient algorithm to solve (\ref{eq:obj})--(\ref{eq:const}), we can
obtain solutions to the Hadamard Maximal Determinant problem.
Given an arbitrary, full-rank, set of $k$ samples of $\mathbf{Y}$, by setting $\tilde{\mathbf{U}}=\mathbf{UY}$, we arrive at the following optimization problem, equivalent to (\ref{eq:obj})--(\ref{eq:const})
\begin{align}
& \underset{\tilde{\mathbf{U}}}{\text{maximize}} 
& & \log | \det \tilde{\mathbf{U}}| \label{eq:hobj} \\
& \text{subject to}
& & \tilde{\mathbf{U}} \in [-1,+1]^{n \times n} \label{eq:hconst}.
\end{align}
Notice that for any value of $n$
\begin{equation}
\underset{\mathbf{w} \in [-1,+1]^{n \times n}}{\mathrm{max}} |\det \mathbf{W}| =
\underset{\mathbf{w} \in \{-1,+1\}^{n \times n}}{\mathrm{max}} |\det
\mathbf{W}|. 
\end{equation}
This is because $\det \mathbf{W}$ is linear in the columns of $\mathbf{W}$ and
so the maximum is obtained at a vertex of $[-1,+1]^{n \times n}$. Thus, we may as well 
consider the maximum over $\{-1, +1\}^n$ instead of 
$[-1, +1]^n$. The optimal $\tilde{\mathbf{U}}$ is the solution to the Hadamard Maximal 
Determinant Problem.
\end{proof}
This observation has many consequences.  Many questions regarding the
Hadamard Maximal Determinant problem have remained open since the problem was
originally posed by Hadamard in 1893 \cite{hadamard1893resolution}.  
Even for moderately sized values of $n$, the
maximum value obtainable by (\ref{eq:hobj}) remains unknown or unverified.  
However, our reduction holds only for $n=k$ and, empirically, the 
program given by (\ref{eq:obj})--(\ref{eq:const}) appears to become easier as 
$k$ grows relative to $n$. Roughly, as we add constraints, we are
removing vertices from the feasible region in a way that leaves vertices
that correspond to solutions. As we show in the next subsection, Algorithm 2 is guaranteed to 
terminate at a vertex, so, removing ``bad'' vertices 
increases the likelihood that we terminate at a vertex that corresponds to a
solution to the blind decoding problem. The following facts are consequences of
this computational equivalence between the Hadamard Maximal Determinant problem
and the blind decoding problem (see for example \cite{macwilliams2006theory}, 
\cite{tressler2004survey}, or \cite{brent2013minors}):
\begin{itemize}
\item For values of $n$ such that Hadamard matrices exists, the global optima to
(\ref{eq:hobj})--(\ref{eq:hconst}) is obtained if and only if $\tilde{\mathbf{U}}$ 
is a Hadamard matrix. 
\item The value of the objective function at vertices of the problem boundary, 
which are the only strict optima of (\ref{eq:obj})--(\ref{eq:const}), 
correspond to the set of possible determinants of$ \{-1,+1\}$-valued
matrices. Understanding this set for general $n$ is an open problem and is considered more difficult than establishing an upper bound 
on the maximum value of the determinant.
\item For $n=k$, there is a one-to-one correspondence between global optima and 
distinct maximal-determinant $[-1,+1]$-valued matrices.
\end{itemize}
Finally, we state the following lemma, which follows directly from the proof of 
Lemma \ref{lem:had_red}:
\begin{lemma}
\label{cor:had_iff}
There will always be a global optimum of (\ref{eq:obj})--(\ref{eq:const}) on a vertex of a feasible region.  If a Hadamard matrix exists,
then all global optima of (\ref{eq:obj})--(\ref{eq:const}) are strict and lie on vertices.
\end{lemma}

\subsection{Behavior of Algorithm \ref{alg:fit_corner}}
\label{sub:behavior}
In this subsection, we show that Algorithm \ref{alg:fit_corner} is guaranteed to
terminate at a vertex of the feasible reason.  This result is important because all solutions to the blind decoding problem will lie on these vertices, as shown in the following claim.
\begin{claim}
\label{clm:vertex}
Solutions to the blind decoding problem lie on vertices on the feasible
region, defined by (\ref{eq:const}).
\end{claim}
\begin{proof}
This is a simple consequence of the fact that $\mathbf{X} \in \{-1,+1\}^{n
\times k}$. Any $\mathbf{U}$ which takes $\mathbf{UY}$ to $\{-1,+1\}^{n \times
k}$ will satisfy exactly $kn$ constraints from (\ref{eq:const}) with equality.
Since the constrained region is given by a polytope with $2kn$ faces, $kn$ of
which are linearly independent, this corresponds with a vertex of the feasible
region. 
\end{proof}
Notice that because (\ref{eq:obj}) is not convex, Claim \ref{clm:vertex} is not immediately obvious, nor is it obvious that either gradient descent or Algorithm \ref{alg:fit_corner} will terminate at a vertex.  We show that there is a small but non-zero chance that gradient descent will not terminate at a vertex. This motivates the study of Algorithm \ref{alg:fit_corner} over Algorithm \ref{alg:fit}.  Concretely, our first result regarding the behavior of Algorithm \ref{alg:fit_corner} is stated as follows:
\begin{theorem}
\label{thm:vertex}
Algorithm \ref{alg:fit_corner} terminates at a vertex of the feasible region of (\ref{eq:obj})--(\ref{eq:const}) with probability 1.
\end{theorem}
The full proof of this theorem is contained in Appendix \ref{apx:theorem_one}.  Here, we sketch the proof of this theorem which will also give the reader intuition as to why the blind decoding problem can, at reasonable dimensions, be practically solved by gradient descent or other similar optimizations methods.

The first step in the proof of Theorem \ref{thm:vertex} is showing that all optima lie on the problem boundary.  This is formally proven in Lemma \ref{lem:fp_boundary}.  This lemma is a simple consequence of the facts that the objective function consists of the composition of a monotonically increasing function (the logarithm) and a multilinear function (the determinant), and that the problem boundary is convex.  These facts imply that, given any point within the feasible region that does not lie on the boundary, we can always move away from the origin in a way that increases the objective function.

We have already established in Lemma \ref{cor:had_iff} that when a Hadamard matrix exists, all optima are strict.  Conversely, at dimensions where Hadamard matrix do not exist, then one can find non-strict optima.  From Lemma \ref{lem:fp_boundary}, we know that these non-strict optima must lie on the boundary of the feasible region. In Lemma \ref{lem:face} and Corollary \ref{cor:nonstrict}, we further characterize these non-strict optima to show that if a non-strict optima exists, then they must be restricted to a linear interval contained on a face of the polytope which defines the feasible region.  We further show that all strict optima, regardless of the existence of a Hadamard matrix must lie on vertices.  We use these fact together with Lemma \ref{lm:max_optima} to guarantee that Algorithm \ref{alg:fit_corner} reaches a vertex.

In Lemma \ref{lm:n_minus_one}, we show how far gradient descent (or Algorithm \ref{alg:fit}) will proceed towards a vertex.  Notice that the constraints in (\ref{eq:const}) act on each row of $\mathbf{U}$ independently, and $\mathbf{U}$ will be at a vertex of the feasible region when there are exactly $n$ constraints active on each row. In fact, we show in Lemma \ref{lm:n_minus_one} that when gradient descent terminates (meaning we have reached an optima), each row of $\mathbf{U}$ will have at least $n-1$ active constraints per row.  

When this occurs, $\mathbf{U}$ will be on an edge of the feasible region; indeed, there is exactly one line on which $\mathbf{U}$ can move while staying on the boundary of the feasible region and not affecting the active constraints.  We can further show that the objective function must be constant along this line.  Thus, for each row with $n-1$ active constraints, we can simply choose a direction at random and move in this direction until we reach a vertex.  We are thus guaranteed that Algorithm \ref{alg:fit_corner} will terminate at the vertex of the feasible region.

\subsection{Maximal Subset Property}
\label{sub:vertex}
We have established that Algorithm 2 always terminates on a vertex of the
feasible region.  However, such a point may either be a global or local optima to (\ref{eq:obj})--(\ref{eq:const}) and may not correspond to a solution to the blind decoding problem.
In this light, we now study when vertices of the feasible region
correspond to solutions to the blind decoding problem and understanding when, if
ever, local optima of (\ref{eq:obj})--(\ref{eq:const}) exist.
We first derive a sufficient condition for the solutions of the blind decoding
problem to correspond to global optima of (\ref{eq:obj})--(\ref{eq:const}).
More precisely, we will study the following condition of the
set $\{ \mathbf{x}_1, \ldots, \mathbf{x}_k \}$: 

\begin{defn}
A matrix $\mathbf{X} \in [-1, +1]^{n \times k}$, and corresponding set $\text{cols}(\mathbf{X}) \subseteq [-1,+1]^{n}$, with $k \geq n$ has the maximal
subset property if there is a subset $\text{cols}(\mathbf{V}) \subseteq \text{cols}(\mathbf{X})$ of size $n$ so that if
$\mathbf{V} \in \mathbb{R}^{n \times n}$ is the matrix with elements of $\text{cols}(\mathbf{V})$ as
columns, then
\begin{equation}
|\det \mathbf{V}| = \underset{\mathbf{W} \in [-1, +1]^{n \times n}}{\mathrm{max}}
           |\det \mathbf{W}|.
\end{equation}
\end{defn}
 That is, $\mathbf{X}$ has the maximal subset property if it
contains a subset of columns that, when viewed as a matrix, has a determinant
that is maximal among all $[-1, +1]$-valued matrices (and hence also all 
$\{-1,+1\}$-valued matrices).  With this definition, we can now state a sufficient 
condition for solutions to our problem to be global optima.
\begin{lemma}
\label{lm:max_optima}
If $\mathbf{X}$ has the maximal subset property, then, for all ATMs $\mathbf{T}$,
 all matrices in the form $\mathbf{U} = \mathbf{TA}^{-1}$ are global optima 
of (\ref{eq:obj})--(\ref{eq:const}).
\end{lemma}
Lemma \ref{lm:max_optima} is proved in Appendix \ref{apx:lemma_three}.  For small $n$, 
we can compute the probability that a set of $k$ samples has the maximal subset property; this result is given in Appendix \ref{sec:rank}.
In Section \ref{sec:empirical} we show that the empirical success probability of Algorithm \ref{alg:fit} with $k$ samples exactly matches the probability distribution derived in Appendix \ref{sec:rank}.

It is natural to ask whether the converse of Lemma \ref{lm:max_optima} is true. In fact,
for $M>2$, an explicit counterexample exists, found through computer simulation,
that shows that the maximal subset property is not a necessary condition. This
is important because the probability of finding a maximal subset through uniform 
sampling decreases rapidly as $M$ and $n$ grow. Indeed, this agrees with our
empirical observations, specifically Table \ref{tb:samples} found in Section \ref{sec:empirical}, which show that the
increase in $k$ required to maintain a constant success probability (assuming
uniform sampling) appears less than quadratic in $M$. Having established the
maximum subset property as a sufficient condition, we now continue our analysis
of the geometry of our non-convex optimization problem by considering specific
values of $n$.
 
\subsection{The Cases $n=2$ and $n=3$}
\label{sub:n23}
For $n=2$ and $n=3$, it can easily be checked that 
all full-rank matrices in $\{-1,+1\}^{n \times n}$ have the maximal subset
property. In other words, the set of possible values of the determinant contains two possible
absolute values, 0 and $2^{n-1}$. For $n=2$, one can verify that ATMs are the only orthogonal
matrices that map elements of this set to other elements of this set.  Further, since a Hadamard matrix exists at $n=2$, all optima are strict and vertices of the feasible region.  These facts imply the following theorem. 
\begin{theorem}
\label{thm:n23iff}
For $n=2$, Algorithms \ref{alg:fit} and \ref{alg:fit_corner} are 
correct (that is, finds a solution to the blind decoding problem) 
if and only if $\mathbf{X}$ has the maximal subset property.
\end{theorem}

For $n=3$, the maximal subset property alone is not sufficient to ensure that Algorithms \ref{alg:fit} or \ref{alg:fit_corner} succeed.  Indeed, for any $\mathbf{X}\in \{-1,+1\}^{3 \times 3}$, there exists a matrix $\mathbf{Q}$ such that $\mathbf{QX} \in [-1,+1]^{3 \times 3}$, $\det \mathbf{Q} = \pm 1$ and $\mathbf{Q} \notin \mathcal{T}$.  This implies the existence of spurious optima whenever $k=3$.  However, these spurious optima will not exist if $k \geq 4$ and $\mathbf{X}$ contains at least one additional distinct column beyond the three required for the maximal subset property.  By a \emph{distinct} column, we mean that the $i$th column of $\mathbf{X}$ is distinct if $\mathbf{x}_i \neq \pm \mathbf{x}_j$ for all $j \neq i$.  Notice that this also implies that all columns of $\mathbf{X}$ are pair-wise linearly independent.  Further, we note that Algorithm \ref{alg:fit} is no longer guaranteed to be correct for $n=3$ because $n=3$ contains no Hadamard matrix.  We now formally state a theorem, proven in Appendix \ref{apx:n3}, regarding the performance of Algorithm \ref{alg:fit_corner}.

\begin{theorem}
\label{thm:n3iff}
When $n=3$, Algorithm \ref{alg:fit_corner} is correct with probability 1 if and only if 
$k \geq 4$ and there exists a $3 \times 4$ matrix $\mathbf{V}$, such that $\text{cols}(\mathbf{V}) \subseteq \text{cols}(\mathbf{X})$, 
$\text{span}(\mathbf{v}_1,\ldots,\mathbf{v}_4) = \mathbb{R}^3$, and all vectors in $\mathbf{V}$ are pair-wise linearly independent. 
\end{theorem}

We now turn our attention to quantifying the probability that the conditions required by Theorems \ref{thm:n23iff} and \ref{thm:n3iff} hold.  
Let
$r(n,k)$ denote the probability that a collection of $k$ vectors in
$\mathbb{F}_2^n$ has rank $n$ (and hence the rank in $\mathbb{R}^n$ must also be
$n$), then the probability of the solver succeeding, given $k$ samples chosen
uniformly at random, is given by:
\begin{equation}
\mathrm{Pr}(\mathrm{Success}) = r(n,k). \label{eq:n23}
\end{equation}
An explicit formula for $r(n,k)$ is derived in Appendix \ref{sec:rank}.  Notice that for $n=2$, equation (\ref{eq:n23}) expresses the probability that the set of global optima contains
solutions to the blind decoding problem. For $n=3$, since we require the existence of a fourth distinct vector, we find that the probability, for a set of $k$ samples chosen uniformly, that all global optima will be solutions to be
\begin{equation}
\mathrm{Pr}(\mathrm{Success}) = r(n,k) \cdot (1-2^{n-k}). \label{eq:n3}
\end{equation}

We note that, for $n=3$, if a collection of samples contains only the MSP and not an additional distinct column, then Algorithm \ref{alg:fit_corner} still has a non-zero probability of finding a solution to the blind decoding problem as the set of global optima still contains the set of all solutions.  Thus the probability of success of Algorithm \ref{alg:fit_corner}, conditioned over a uniform selection of samples, is bounded between (\ref{eq:n23}) and (\ref{eq:n3}).
In Section \ref{sec:empirical} we compare these distributions to our empirical results.

\subsection{The case $n=4$.}
\label{sub:n4}
At dimension $n=4$, the problem geometry gets slightly more complicated.  The
set of possible values of the determinants of $\{-1,+1\}^{4 \times 4}$ increases to
$\{0,\pm 8, \pm 16\}$, which means that not all non-singular vertices of (\ref{eq:const}) are global
optima to (\ref{eq:obj})--(\ref{eq:const}).  However, we show that for $n=4$, the only optima of (\ref{eq:obj})--(\ref{eq:const})
are indeed global optima.  Unfortunately, for $n=k$, not all global optima are solutions to the blind decoding problem. Nonetheless, we are able to show that for $n=4$, Algorithms 1 and 2
both succeed (and solve the blind decoding problem) with probability 1 under proper input conditions.   

Before stating Theorem \ref{thm:n4iff}, which is proved in Appendix \ref{apx:theorem_three}, we must also introduce equivalence classes of 
Hadamard matrices. We say that two Hadamard matrices $\mathbf{H}_1$ and $\mathbf{H}_2$ are equivalent if there exists and ATM $\mathbf{T}$ such that $\mathbf{H}_1 = \mathbf{T \cdot H}_2$.  This is an equivalence relation, and thus decomposes the set of Hadamard matrices into equivalence classes.  For $n=4$, there are exactly two
equivalence classes, which we denote as $\mathcal{H}_4^{(1)}$ and
$\mathcal{H}_4^{(2)}$, that are defined as follows:
\begin{align}
\mathcal{H}_4^{(1)} &= 
\left\{ 
\mathbf{T}
\begin{bmatrix}
1 &  1 &  1  &  1 \\ 
1 &  1 &  -  &  - \\
1 &  - &  1  &  - \\
1 &  - &  -  &  1  
\end{bmatrix}, 
\forall \mathbf{T} \in \mathcal{T}
\right\},  \\
\mathcal{H}_4^{(2)} &= 
\left\{
\mathbf{T}
\begin{bmatrix}
 - &  1 &  1  &  1 \\ 
 1 &  - &  1  &  1 \\
 1 &  1 &  -  &  1 \\
 1 &  1 &  1  &  - 
\end{bmatrix},
\forall \mathbf{T} \in \mathcal{T}
\right\},
\end{align}
where ``$-$'' denotes $-1$.
Notice that all vectors in $\mathbb{F}_2^4$ appear as column vectors in either
$\mathcal{H}_4^{(1)}$ or $\mathcal{H}_4^{(2)}$.
We say that a vector belongs to an equivalence class if it appears as a column
vector in that equivalence class.
We now state our result for the case $n=4$, which is proven in Appendix \ref{apx:theorem_three}. Notice that because a Hadamard matrix exists for $n=4$, then by
Lemma \ref{cor:had_iff}, the only optima of (\ref{eq:obj})--(\ref{eq:const})
are strict and hence gradient descent will always terminate at a vertex even
without the modification given in Algorithm \ref{alg:fit}.
\begin{theorem}
\label{thm:n4iff}
When $n=4$, Algorithm \ref{alg:fit} is correct with probability 1 if and only if $k \geq 5$ and $\text{cols}(\mathbf{X})$
contains at least four linearly independent vectors from $\mathcal{H}_4^{(i)}$
and a fifth vector from $\mathcal{H}_4^{(j)}$ for $i \neq j$. 

Algorithm \ref{alg:fit} will be correct with probability 0.5 if $\text{cols}(\mathbf{X})$ has only four
linearly independent vectors belonging to the same equivalence class.
\end{theorem}

Theorem \ref{thm:n4iff} implies that we will always require at least 5 samples in order to solve the blind decoding problem.  Further, assuming that the source symbols are chosen uniformly at random, this result allows us to quantify the success probability of the blind decoding algorithm.  This is done in Appendix \ref{apx:theorem_three}, where we show that the success probability for $n=4$ is given by:
\begin{equation}
\text{Pr(Success)} = r(4,k) \cdot r(3,k)^4 \cdot (1-2^{n-k}).
\end{equation}

\subsection{Larger $n$}
\label{sub:large_n}
In this subsection we discuss the performance of Algorithm \ref{alg:fit} for
larger values of $n$.  
In Figure \ref{fig:finding_hadamard}  we use Algorithm \ref{alg:fit} to attempt to
find maximal determinant matrices, as described in Lemma \ref{lem:had_red}.
For $1 \leq n \leq 5$, Algorithm \ref{alg:fit} terminated at a global maximum 
100\% of the time, supporting the claim that there are no local maxima in these cases, as 
explicitly proven for dimension 1,2,3, and 4.  This also suggests that a similar
theoretical guarantee may exist for $n=5$, but proving such a result in the
same manner as used for the case $n=4$ would be computationally expensive.

For dimensions $n>5$, such an analysis seems extremely difficult.
Indeed, for even reasonably small values of $n$, the set of possible 
determinants of $\{-1,+1\}$-valued matrices is not well understood, and for very
large values of $n$ the maximal value of the determinant is only known for special cases of $n$: see, for example, \cite{brent2013minors}.

We can however compare Figure \ref{fig:finding_hadamard} with results obtained in Section \ref{sec:empirical} (notably Figure \ref{fig:empirical_only}).
Despite the fact that the odds of finding a global optima
decreases when $n$ grows with $k=n$, the probability of success of Algorithm
\ref{alg:fit} empirically grows toward 1 when $k$ is sufficiently large.  
Intuitively, this happens because adding additional constraints removes vertices from the
feasible region.  This has the effect of removing both local optima as well as global optima that do not correspond to solutions to our problem. Based on the theory established in this section and empirical
results from Section \ref{sec:empirical}, we make the following
conjecture about the behavior of Algorithms \ref{alg:fit} and
\ref{alg:fit_corner} for general values of $n$. 
\begin{conj}
If $k$ is slightly larger than $n$, and if $\mathbf{X}$ is selected uniformly 
from the set of all $n \times k$ matrices that have the maximal subset property, 
then with high probability, the only optima in (\ref{eq:obj})--(\ref{eq:const}) 
are $\mathbf{U}=\mathbf{TA}^{-1}$, for all ATMs $\mathbf{T}$.
\end{conj}

\begin{figure} 
    \begin{centering}
\begin{tikzpicture}
\begin{axis}[
	ymin=0, ymax=1,
    ytick={0,0.1,...,1.1},
    xmin=0, xmax=16,
	ylabel=Success Probability,
    xlabel=$n$,
    height=\columnwidth,
    width=\columnwidth
]
\addplot [blue!90,mark=*,mark options={draw=blue!75,fill=blue!40}] coordinates
{(2,1) (3,1) (4,1) (5, 1) (6, 0.952634695) (7,0.84758118)
(8,0.810136157) (9,0.741197183) (10,0.379654859) (11,0.62)
(12,0.297) (13,0.372139303) (14,0.071) (15,0.006)};

\end{axis}
\end{tikzpicture}
    \par\end{centering}
    \caption{The probability that Algorithm \ref{alg:fit_corner} finds a
maximal-determinant matrix when used as described in Lemma
\ref{lem:fp_boundary}, averaged over 1000 samples. 
For dimension at most 5, all optima are global.  Above
this, we are not guaranteed to find terminate at a maximal matrix.  However, as
$k$ grows relative to $n$, the probability increases again.  See Figure
\ref{fig:analytic}.}
    \label{fig:finding_hadamard}
\end{figure}
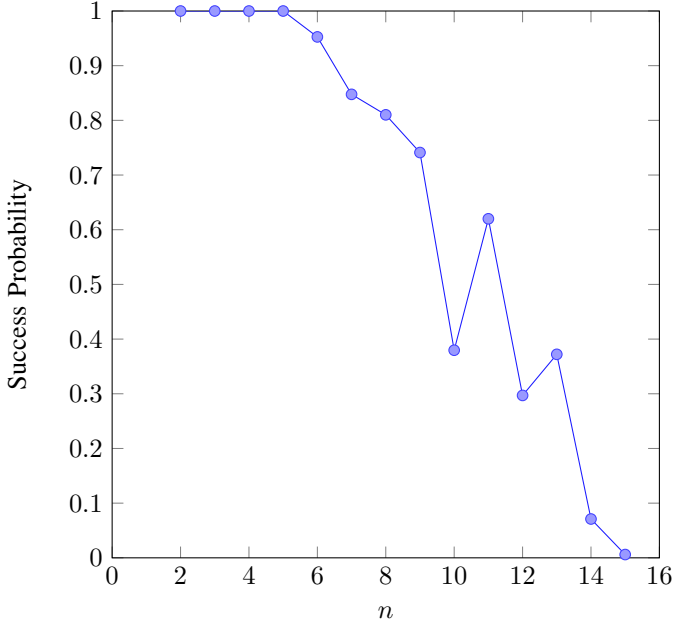

\section{Empirical Results}
\label{sec:empirical}
Having established theoretical results regarding the correctness of our algorithm, we now turn our attention to empirical results.  The simulation results contained in this section are entirely based on Algorithm 1 and demonstrate that Algorithm 2 is unnecessary in practice, at least for low dimensions. The empirical performance of Algorithm 2 does not noticeably improve over the performance of Algorithm 1.  
In order to assess the performance of Algorithm 1, we constructed two sets of
experiments.  In the first, we ran Algorithm 1 for various values of $n$ and
$M$ without channel noise in order to empirically test the conditions under which 
the solver will return the correct solution.  In the second, we ran the
algorithm using realistic channel conditions and compared the results to the
Zero-Forcing and Maximum-Likelihood decoders, both with perfect and imperfect
CSI.

\begin{figure} 
    \begin{centering}
    \vspace{-10mm}
    \includegraphics[width=0.96\columnwidth]{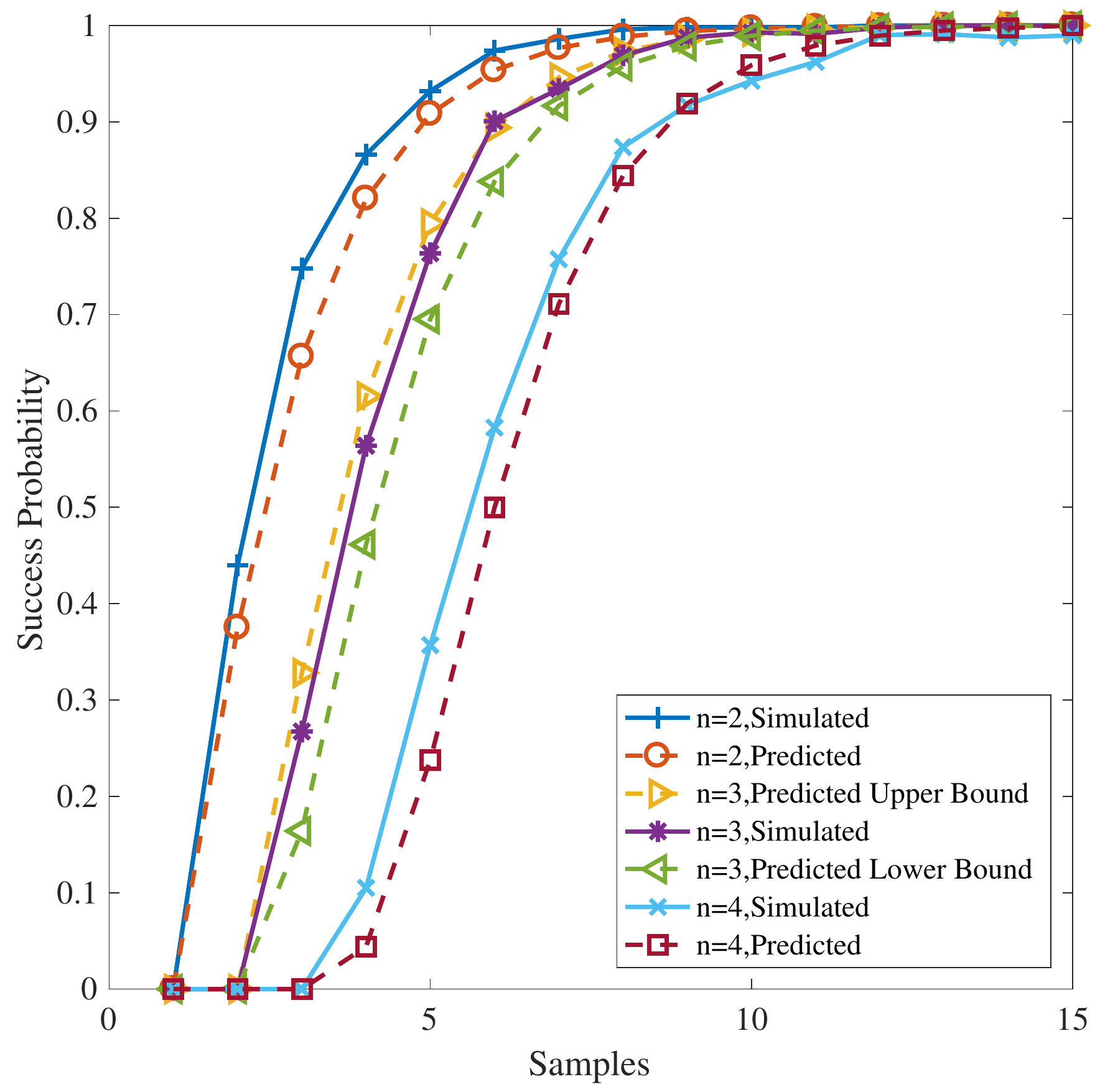} 
    \par\end{centering}
    \caption{Success probability of Algorithm 1 versus number of samples for 
    $n=2,3,4$ and
$M=2$.  Simulations were run over 200 trials. The predicted results are the probability 
that $\mathbf{X}$ has the requisite subset of columns to ensure correctness of Algorithm 1, as discussed in Section V and Appendix \ref{sec:rank}. }
    \label{fig:analytic}
\end{figure}

\begin{figure} 
    \begin{centering}
    \includegraphics[width=1\columnwidth]{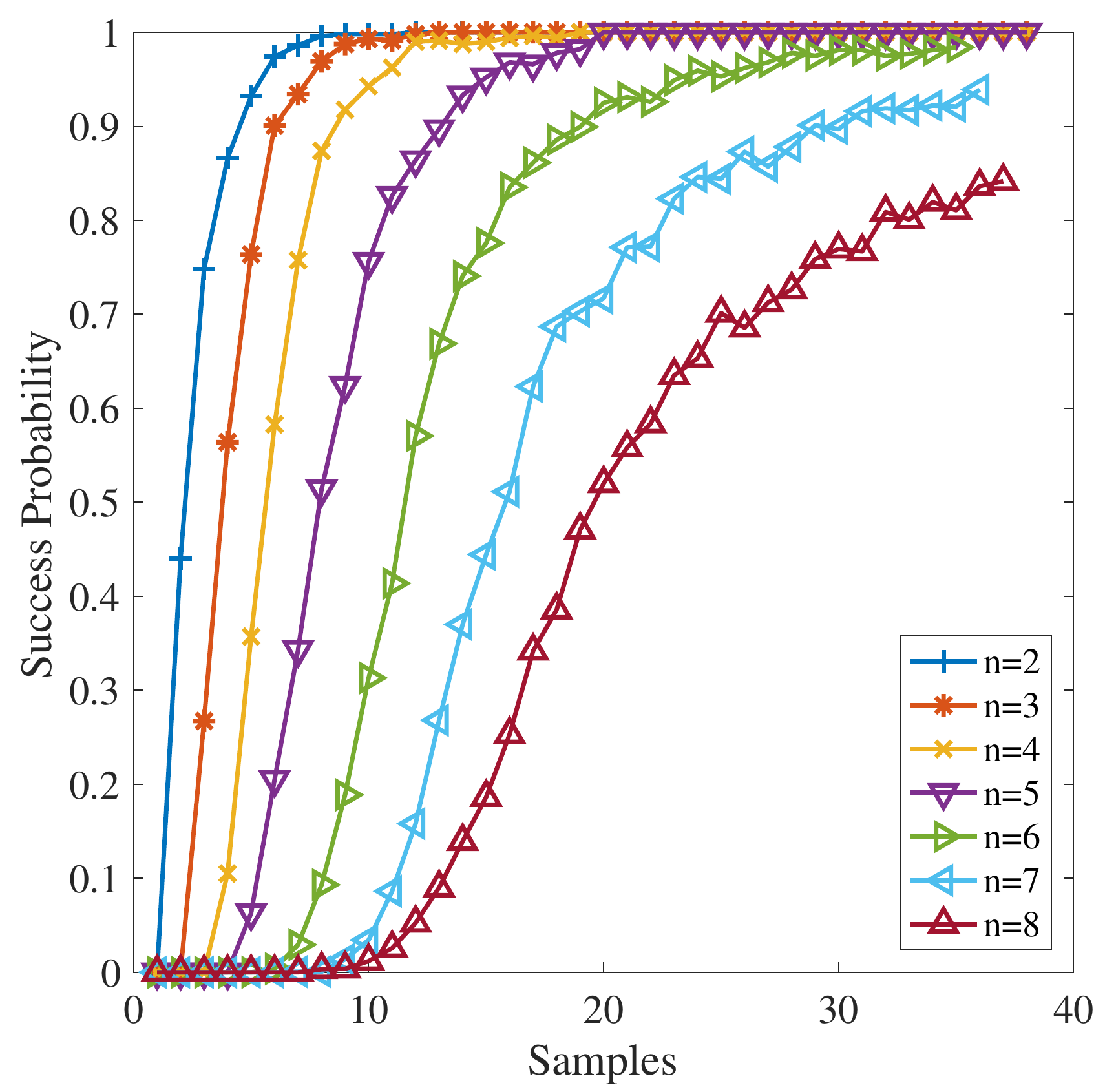} 
    \par\end{centering}
    \caption{Empirical success probability of Algorithm 1 for $n=2,\ldots,8$ and $M=2$ with no AWGN.  Simulations were run over 200 trials. Beginning at $n=6$, local optima exist when $n=k$ and $\mathbf{X}$ contains the maximum subset property.  However, as $k$ increases these local optima may become infeasible, increasing the success probability of Algorithm 1.}
    \label{fig:empirical_only}
\end{figure}

\begin{table}
\caption{Sample Size Requirements}
\label{tb:samples}
\begin{center}
\begin{tabular}{|c||c|c|c|c|}
\hline
\multicolumn{5}{|c|}{Algorithm 1} \\ \hline
n & M=2 & M=4 & M=8 & M=16 \\ \hline
2 & 5 & 14 & 29 & 56 \\ \hline
3 & 6 & 23 & 45 & 87 \\ \hline
4 & 10 & 32 & 60 & 118 \\ \hline
5 & 14 & 42 & 98 & 150 \\ \hline
%
\hline
\multicolumn{5}{|c|}{ Algorithm presented in \cite{talwar1996blind}} \\ \hline
n & M=2 & M=4  & M=8    & M=16 \\ \hline
2 & 5   & 33   & 182    & 913 \\ \hline
3 & 13  & 182  & 2,006   & 20,326 \\ \hline
4 & 33  & 913  & 20,326  & 416,140 \\ \hline
5 & 79  & 4,369 & 196,711 & 8,111,980 \\ \hline
\end{tabular}
\end{center}
\vspace{2mm}
The table at the top shows the number of samples required for various values of 
$n$ and $M$ to recover $\mathbf{U}$ in the correct form with 90\% success rate
using Algorithm \ref{alg:fit}.  The table at the bottom represents the number of
samples needed to ensure a 90\% success rate using either the ILSP or the ILSE
techniques presented in \cite{talwar1996blind}.
\end{table}

Table \ref{tb:samples} summarizes the number of samples required for various values of
$n$ and $M$ so that Algorithm \ref{alg:fit} has a 90\% probability of returning
an optimal solution to (\ref{eq:obj})--(\ref{eq:const}). For the values of $n$ presented, the success probability is almost entirely conditioned upon the input value of $\mathbf{X}$ rather than randomness in Algorithm \ref{alg:fit}; that is, running Algorithm \ref{alg:fit} multiple times on the same $\mathbf{X}$ will not improve success
rates.

Figure \ref{fig:analytic} shows the expected success rate for $n=2,3,4$ which is
based on the theory in Section \ref{sec:prelim} and Appendix \ref{sec:rank}. 
The results in this plot are for the case $M=2$ which corresponds to Binary Phase Shift Keying (BPSK) in the absence of noise.
For $n=2$, the theoretical success probability is the probability that $\mathbf{X}$ has the maximal subset property.  
For $n=3$ and $n=4$, success is only guaranteed if $\mathbf{X}$ has the maximal subset property as well as at least one additional distinct vector.  
For $n=3$, the expected success rate of Algorithm \ref{alg:fit} is not known when $\mathbf{X}$ has the maximal subset property alone.  
The probability that $\mathbf{X}$ has the maximal subset property is plotted as a lower bound on performance in this case; the upper bound given in Figure \ref{fig:analytic} expresses the probability that $\mathbf{X}$ has the maximal subset property as well as one additional vector.  
For $n=4$, as shown in Section \ref{sec:prelim}, we know that Algorithm \ref{alg:fit} will succeed with probability 0.5 when $\mathbf{X}$ has the maximal subset property alone; this is reflected in the theoretical prediction for this case. 
We note that for $n=2,3,$ and $4$, the empirical observations match the expected theoretical performance.

Figure \ref{fig:empirical_only} shows the empirical success probability of Algorithm \ref{alg:fit} up to $n=8$.  
This plot demonstrates two important features regarding the performance of Algorithm \ref{alg:fit} as $n$ grows.  
First, for $n > 5$, it is know that local optima may exist.  
Figure \ref{fig:finding_hadamard} from Section \ref{sec:prelim} gives the probability that when $n=k$ and $\mathbf{X}$ has the maximal subset property, Algorithm \ref{alg:fit} will find a global optima.  
However, we can see in Figure \ref{fig:empirical_only} that for large enough values of $k$, the success probability of Algorithm \ref{alg:fit} exceeds this probability.
This is because these additional samples cause local optima to become infeasible, increasing the probability Algorithm \ref{alg:fit} will find a global optima.
Additionally, these results show that the required values of $k$ appear to scale favorably as $n$ grows.  We further note that $n \leq 8$ captures nearly all MIMO systems found in use today.

Figures \ref{fig:est_error} and \ref{fig:bpsk_est_error} shows the symbol error
rate performance of the blind decoder compared to standard MIMO decoding
algorithms.  Figure \ref{fig:est_error} gives an example with high SNR and high
modulation order, with the parameters $n=3, M=32, c=3$, while Figure
\ref{fig:bpsk_est_error} shows the case $n=4, M=2, c=3$ at SNR values typically
found in modern cellular systems.  Despite having less side information,
the performance of the blind decoder (Algorithm \ref{alg:fit}) is only slightly
worse than the ZF and ML decoders with perfect CSI; there appears to be less
than 3 dB loss associated with the blind decoder.  The simulation used a fading block
length of 400 samples, and ran over a total of 500 fading blocks per SNR. 
In high dimensions, large numbers of constraints leads to numerical
instability, requiring the step size to be extremely small, and making the 
solver slow to converge.  Improving the runtime of our algorithm, for example
through an intelligent selection of a subset of received samples, is a topic 
of future research.

\begin{figure} 
    \begin{centering}
	\includegraphics[width=0.95\columnwidth]{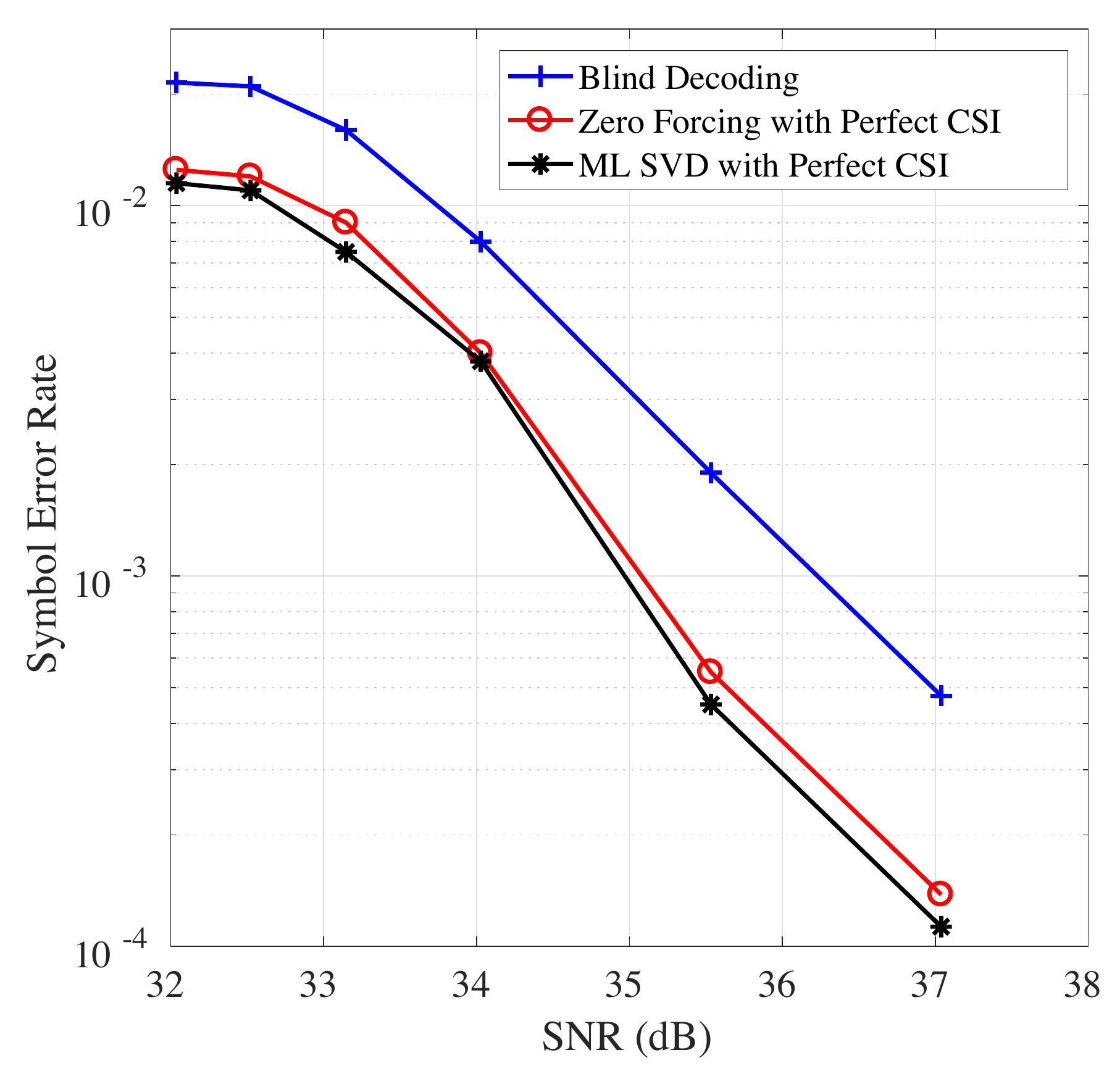} 
    \includegraphics[width=0.95\columnwidth]{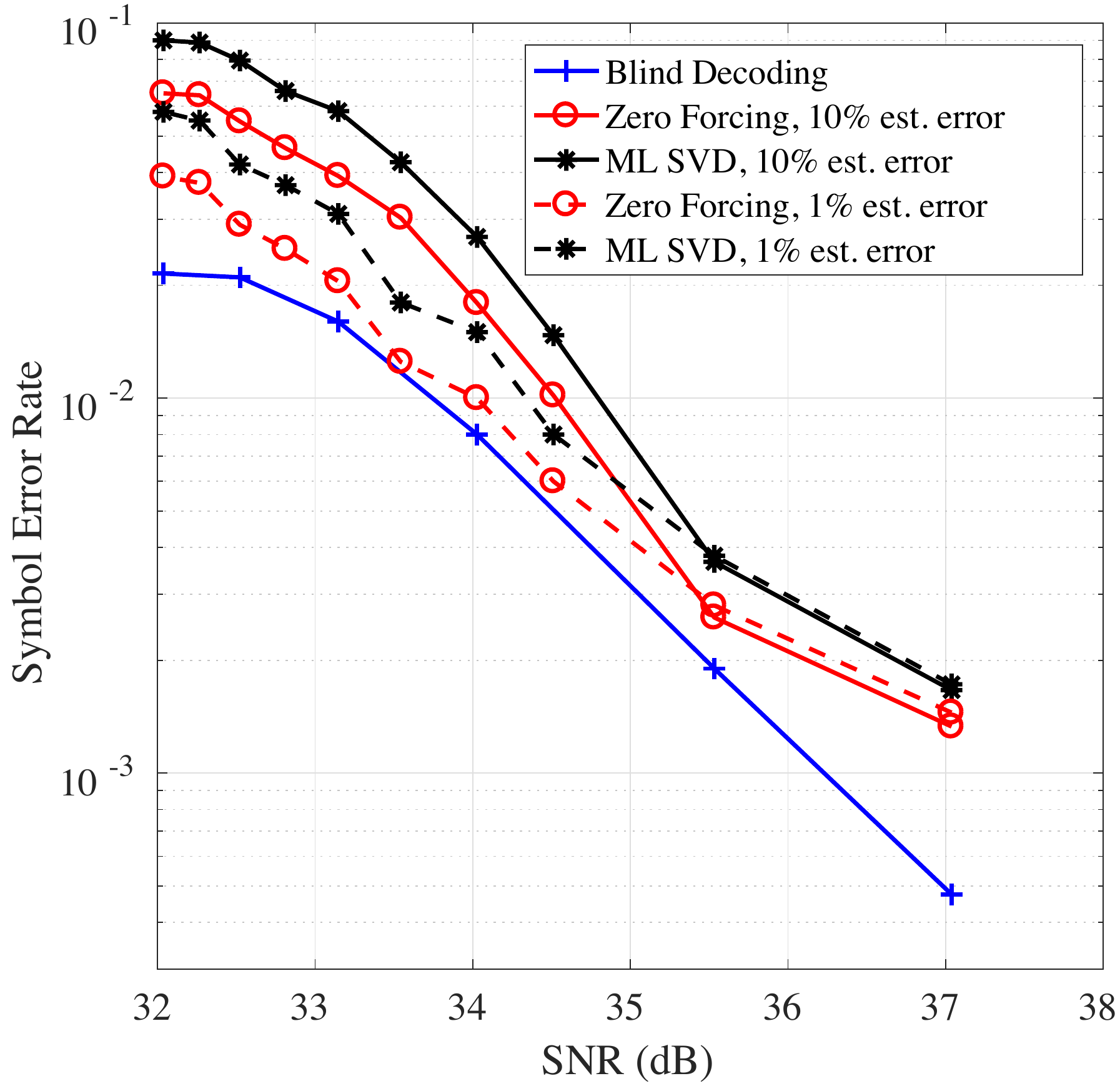}
    \par\end{centering}
    \caption{Decoding performance for $n=3, M=32, c=3$ compared to the
zero-forcing and maximum-likelihood decoder, implemented through parallel
channel decomposition.  The top figure has no estimation error present, 
the bottom figure has compares blind decoding to both ML and ZF decoding with
imperfect CS.  These figures show that if the error in the CSI is even one percent
of the channel noise variance, then blind decoding outperforms both these
algorithms. The blind decoding algorithm appears to have at most a 3dB
loss over decoding with perfect CSI.}
    \label{fig:est_error}
\end{figure}

\begin{figure} 
    \begin{centering}
	\includegraphics[width=1\columnwidth]{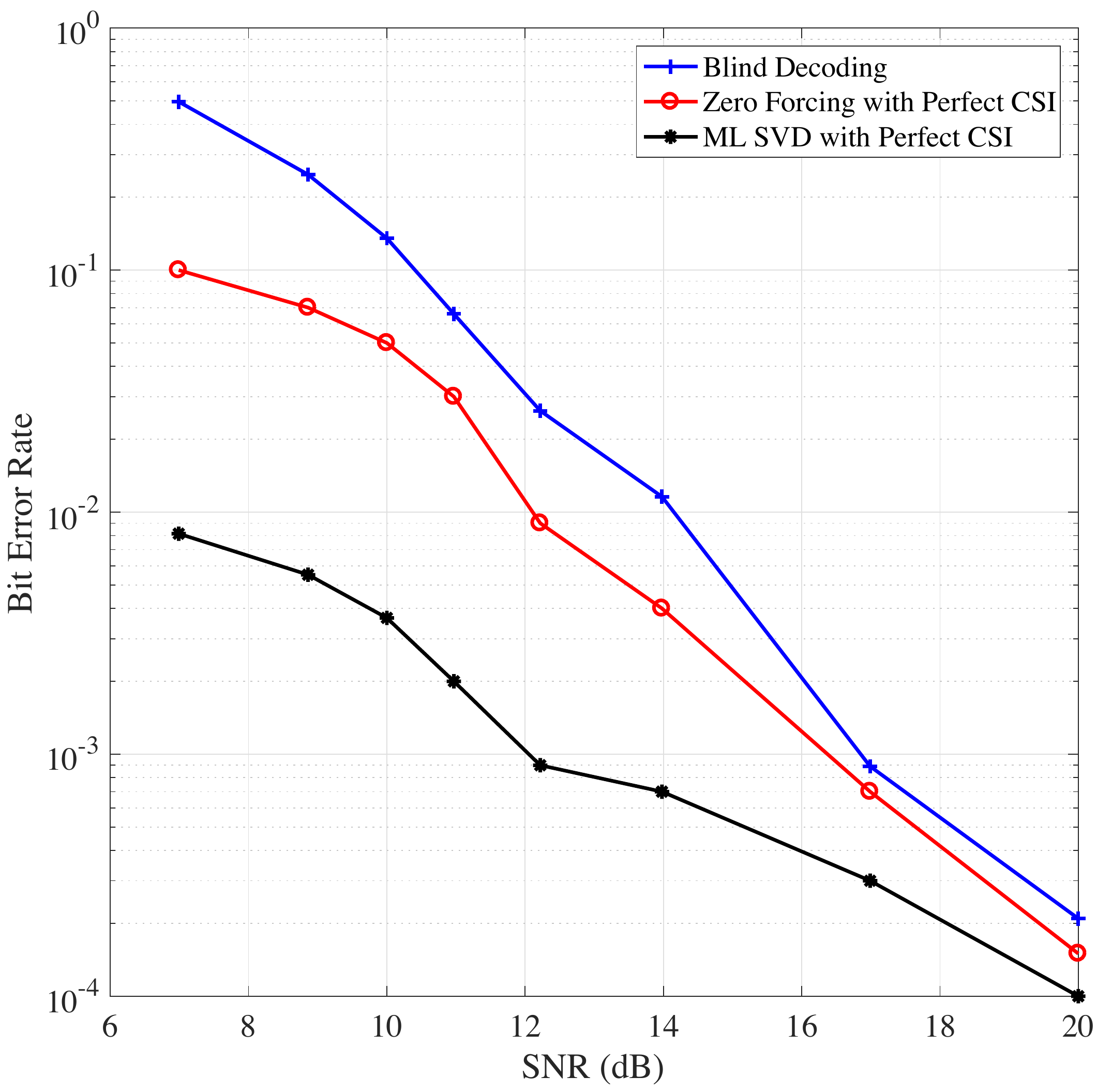} 
    \includegraphics[width=0.99\columnwidth]{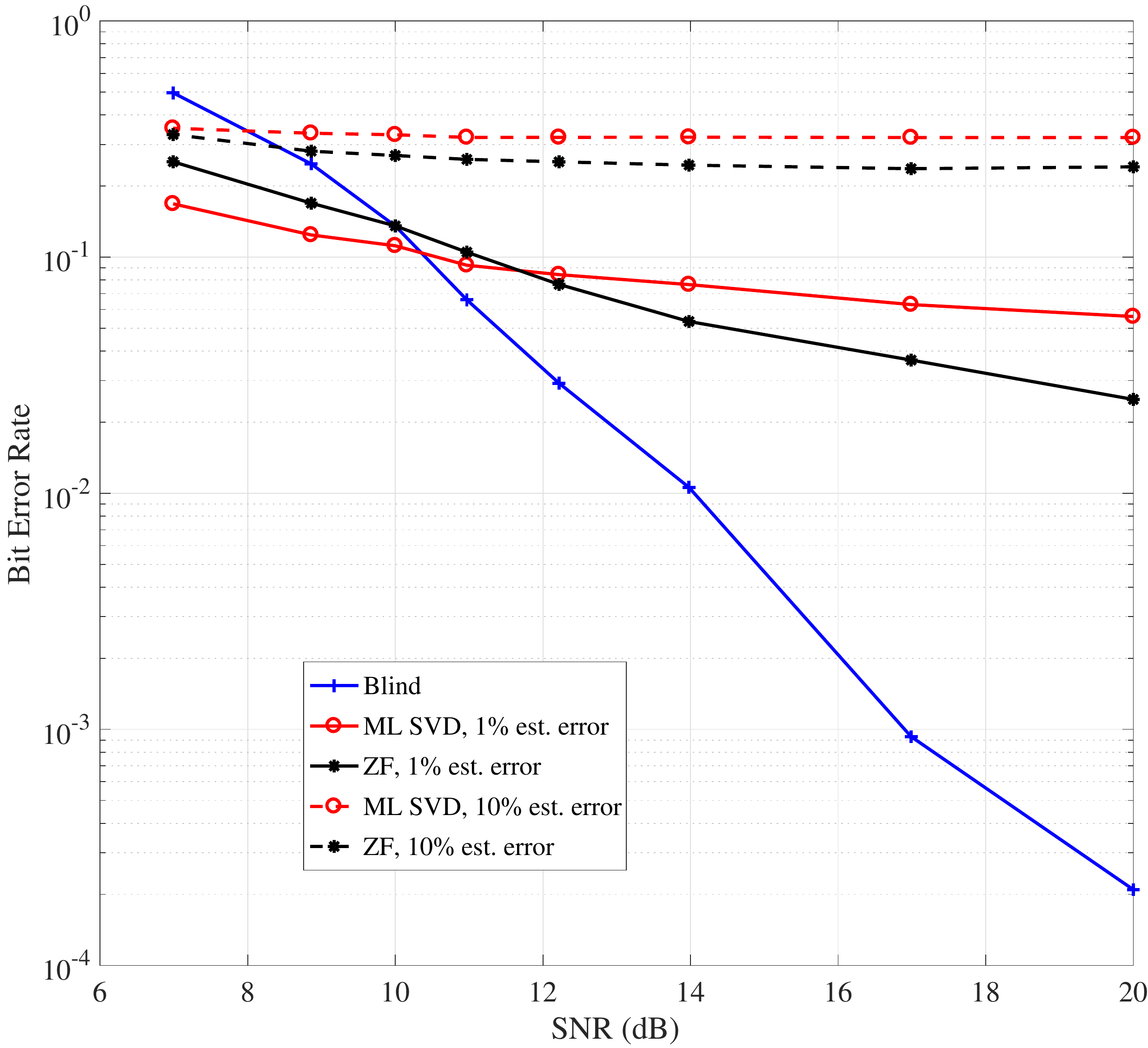}
    \par\end{centering}
    \caption{Decoding performance for $n=4, M=2, c=3$.  The decoding
performance is similar in comparison to standard MIMO algorithms as 
the $n=3, M=32$ case shown in Figure \ref{fig:est_error}.
}
    \label{fig:bpsk_est_error}
\end{figure}
Motivated by real-world usage, we compared blind decoding (Algorithm
\ref{alg:fit}) to the ZF and ML decoders
with imperfect CSI.  If we assume that the channel is estimated through a set of 
known pilot symbols that will be corrupted by Gaussian noise, the error in the 
CSI will be i.i.d. Gaussian.  This is a realistic assumption in most wireless
systems, and the model we used in our simulations.
In Figures \ref{fig:est_error} and \ref{fig:bpsk_est_error}, we also plot the performance of Algorithm
\ref{alg:fit} against ZF and ML when the variance of the error in the channel gain matrix 
is 1\% of that of the AWGN in the channel and for 10\% estimation error.  In both cases, 
Algorithm \ref{alg:fit} significantly outperforms the ZF and ML decoders.

\section{Conclusion}
\label{sec:conc}
We have provided an algorithm to jointly estimate MIMO channels and decode the
underlying transmissions in block fading channels.  This algorithm performs
gradient descent on a non-convex
optimization problem.  Empirically, this algorithm has a performance
loss on the order of several decibels versus schemes with perfect CSI, 
but its performance becomes
superior when CSI knowledge is imperfect.  This algorithm is
practical in that it works well for block-fading channels with realistic
coherence times. In addition to the important application of decoding in MIMO
channels with missing or imperfect CSI, our algorithm is potentially useful 
from the point of view of an eavesdropper who does not know the pilot symbols 
but is trying to recover $\mathbf{x}$. 

We present in-depth analysis of the geometry of this non-convex optimization
problem. Specifically we prove that for $M=2$, and small values of $n$, all 
optima are global and give necessary and sufficient conditions for when these
optima correspond to solutions to the blind decoding.  For general values of $n$, we
relate the problem to the Hadamard Maximal Determinant problem and give evidence
that providing matching theoretical guarantees for larger values of $n$ is
likely difficult.  However, our empirical results suggest that our algorithm
remains feasible for values of $n$ commonly found in modern MIMO systems.


This paper also motivates a suite of open theoretical problems related to the
performance of our algorithm. For example, we provide no theoretical results
that analytically explain the performance in the presence of AWGN, and for
$M>2$. We also leave open possible extensions to rectangular or complex-valued
channels, as well as more efficient algorithms than gradient descent that solve
the blind decoding problem.

\appendices
\section{Proof of Proposition \ref{prop:fullrank}}
\label{apx:prop_one}
In this appendix, we prove the result of Proposition \ref{prop:fullrank} shows that solutions to (\ref{eq:obj})--(\ref{eq:const}) are meaningful if and only if the channel gain matrix is full rank.
\begin{customprop}{\ref{lm:max_optima}}
If the matrix $\mathbf{Y}$ is not full rank, then (\ref{eq:obj})--(\ref{eq:const}) 
is unbounded above. Conversely, if $\mathbf{Y}$ is full rank and $k \geq n$, then
(\ref{eq:obj})--(\ref{eq:const}) is bounded above and feasible.  
\end{customprop}
\begin{proof}
If $\mathbf{Y}$ is not full rank, then there is some nonzero vector $\mathbf{v}
\in \mathbb{R}^n$ s.t. $\mathbf{v}^T \mathbf{Y} = \mathbf{0}$. Then the matrix
\begin{align}
\mathbf{UY} &=
c_1 \cdot \left( \mathbf{A}^{-1} + c_2 \cdot
\begin{bmatrix}
\horzbar & \mathbf{v}^T & \horzbar \\
			 \\
	&  \mathbf{0} &   \\
			 \\
\end{bmatrix} \right) \mathbf{Y} \nonumber \\
&=c_1 \mathbf{X} + c_1 \mathbf{A}^{-1} \mathbf{e}
\end{align}
satisfies (\ref{eq:const}), for some value of $c_1$, but the objective function 
(\ref{eq:obj}) grows without bound as $c_2$ grows.

Conversely, if $\mathbf{Y}$ is full rank and $k \geq n$, then the left nullspace
of $\mathbf{Y}$ is $\{\mathbf{0}\}$. Thus, any non-zero row in $\mathbf{U}$ will always
affect (\ref{eq:const}) and the maximum $\| \mathbf{u}_i \|_2$ will be bounded
above for all $i$, implying that (\ref{eq:obj}) is bounded above.  Similarly, 
for any $\mathbf{Y}$, there will always be $\mathbf{U}$ that satisfies
(\ref{eq:const}), for example, consider $\mathbf{U}=0$.
\end{proof}
\section{Proof of Lemma \ref{lm:max_optima}}
\label{apx:lemma_three}
\begin{customlemma}{\ref{lm:max_optima}}
If $\mathbf{X}$ has the maximal subset property, then, for all ATMs $\mathbf{T}$,
 all matrices in the form $\mathbf{U} = \mathbf{TA}^{-1}$, are global optima 
of (\ref{eq:obj})--(\ref{eq:const}).
\end{customlemma}
Lemma
\ref{lm:max_optima} follows from the following claim.
\begin{claim}
\label{lm:detconst}
Suppose that $\mathbf{X}$ has the maximal subset property.  Then for all matrices 
$\mathbf{M}$ such that $\|\mathbf{Mx}\|_\infty \leq 1$ for all $\mathbf{x} \in
\mathbf{X}$, we have $| \det \mathbf{M} | \leq 1$.
\end{claim}
\begin{proof}
Let $\mathbf{x}_1,\ldots,\mathbf{x}_n$ be the set guaranteed by the maximal subset property.  
Let $\mathbf{V}$ be the matrix whose columns are
$\mathbf{x}_1,\ldots,\mathbf{x}_n$, so $|\det \mathbf{V}|$ is maximal.  
If $|\det \mathbf{M}| > 1$, then $|\det \mathbf{MV} | > |\det \mathbf{V}|$.   
This would imply that $\mathbf{MV}$ cannot be
contained in $[-1, +1]^{n \times n}$.
\end{proof}
\noindent Thus, if the matrix $\mathbf{X}$ has the maximal subset property, then the
optimal solution to (\ref{eq:obj})--(\ref{eq:const}) have $|\det
\mathbf{UA}|=1$, and hence all $\mathbf{U}=\mathbf{TA}^{-1}$ correspond to
global optima. This completes the proof of Lemma \ref{lm:max_optima}.
\section{Proof of Theorem \ref{thm:vertex}}
\label{apx:theorem_one}
\begin{customthm}{\ref{thm:vertex}}
Algorithm \ref{alg:fit_corner} terminates at a vertex of the feasible region
with probability 1.
\end{customthm}
We begin the proof of Theorem \ref{thm:vertex} with a lemma that is a simple 
consequence of the fact that the determinant is a
multilinear function. 
\begin{lemma}
\label{lem:fp_boundary}
All optima of (\ref{eq:obj})--(\ref{eq:const}) lie on the boundary of the feasible region.
\end{lemma}
\begin{proof}
Consider a feasible, full-rank $\mathbf{U}$ which is not on the boundary of the
feasible region, i.e. $\|\mathbf{Uy}_i\|_\infty < 1$ for all $i \in [1,\ldots,k]$. Suppose
$|\mathbf{u}^{(j)} \mathbf{y}_i| = c$ for some $i,j$ and some $0 < c < 1$. If we set
\begin{equation}
\tilde{\mathbf{u}}^{(j)} = (1+\epsilon) \mathbf{u}^{(j)}
\end{equation}
for some $0 < \epsilon \leq c$,
and form $\tilde{\mathbf{U}}$ from the matrix $\mathbf{U}$ by replacing row 
$\mathbf{u}^{(j)}$ with 
$\tilde{\mathbf{u}}^{(j)}$, then, because the determinant is linear in the
rows of $\mathbf{U}$, we have: 
\begin{equation}
| \det \tilde{\mathbf{U}} | = (1+\epsilon) | \det \mathbf{U} | > | \det 
\mathbf{U} |
\end{equation}
and $\tilde{\mathbf{U}}$ is still feasible.
\end{proof}
In other words, if $\mathbf{U}$ is not on the problem boundary, the multilinearity of the
determinant function implies that we can always move towards the problem
boundary (and away from the origin) in a way that increases the objective
function.

We can also use the fact that the feasible region is formed by an
$n$-dimensional polytope to further categorize the optima of our optimization
problem, as stated in the following lemma.
Let $\mathcal{P}$ denote the polytope that describes the feasible region given by (\ref{eq:const}), and let $\mathcal{F}$ denote a face of this polytope.  Since each row of the matrix $\mathbf{U}$ acts on the constraints in an independent manner, we say that a row is ``active'' if there are $n$ linearly independent constraints active on this row.  If all rows of $\mathbf{U}$ are active then $\mathbf{U}$ is a vertex of $\mathcal{P}$; further, if there are $i$ active rows, then $\mathbf{U}$ lies on a face of dimension at most $n (n - i)$.  

\begin{lemma}
\label{lem:face}
Suppose that $\mathbf{U}$ is in the interior of a face $\mathcal{F}$.  Then there exists a $\mathbf{v}\in \mathbb{R}^n, i \in [n], \lambda_1 < \lambda_2 \in \mathbb{R}$ such that the interval $\mathbf{I}$ defined as 
\begin{equation}
\label{eq:int}
\mathbf{I} = \{\mathbf{U} + \lambda \mathbf{e}_i \mathbf{v}^\intercal |\, \lambda \in [ \lambda_1, \lambda_2 ] \}
\end{equation}
satisfies $\mathbf{I} \subseteq \mathcal{F}$.
Further, the points given by $\mathbf{U} + \lambda_j \mathbf{e}_i \mathbf{v}^\intercal$ for all $j = 1,2$ lie on a face of lower dimension than $\mathcal{F}$.
\end{lemma}
\begin{proof}
The polytope $\mathcal{P}$ is bounded by $| \langle \mathbf{u}^{(j)}, \mathbf{y}^{(i)} \rangle | \leq 1$ for $\mathbf{y}^{(1)}, \ldots, \mathbf{y}^{(k)}$ and all rows $\mathbf{u}^{(j)}$ of $\mathbf{U}$. If $\mathbf{U}$ is not at a vertex, then there exists a row $\mathbf{u}^{(j)}$ that is not active.  For this non-active row $\mathbf{u}^{(j)}$, say that $\Omega = \{\mathbf{y}^{(i)}\; |\; |\langle \mathbf{u}^{(j)}, \mathbf{y}^{(i)} \rangle | = 1 \}$.  We must have $\dim (\Omega) < n$.  Thus, there exists a $\mathbf{v} \in \Omega^\perp$ and a $\mathbf{y}^{(i)} \notin \text{cols}(\mathbf{Y})$ such that $\langle \mathbf{v}, \mathbf{y}^{(i)} \rangle \neq 0$.  This is true since if, for all $\mathbf{y}^{(i)} \in \mathbf{Y}, \langle \mathbf{v}, \mathbf{y}^{(i)} \rangle = 0$, then $\mathbf{Y}$ is not full rank.  We now require the following claim:
\begin{claim}
For small enough $\lambda$, $\mathbf{U} + \lambda \mathbf{e}_j \mathbf{v}^\intercal \in \mathcal{F}$.
\end{claim}
\begin{proof}
The quantity $\lambda \mathbf{e}_j \mathbf{v}^\intercal$ only affects constraints acting on row $j$ of $\mathbf{U}$, and for all $\mathbf{y}^{(i)} \in \Omega, (\mathbf{U} + \lambda \mathbf{e}_j \mathbf{v}^\intercal) \cdot \mathbf{y}^{(i)} = \mathbf{Uy}^{(i)}$.  Thus, no active constraints have been affected which implies that we have not left the face $\mathcal{F}$.  
\end{proof}
It remains to show that for appropriate values of $\lambda_1$ and $\lambda_2$, the points given by $\mathbf{U} + \lambda_j \mathbf{e}_i \mathbf{v}^\intercal$, for all $j = 1,2$, lie on a lower dimensional face than $\mathcal{F}$.  Since there exists a $\mathbf{y}^{(i)}$ such that $\langle \mathbf{v}, \mathbf{y}^{(i)} \rangle \neq 0$, there must be exactly two values $\lambda_1, \lambda_2$ such that $\langle \mathbf{u}^{j}+\lambda_i \mathbf{e}_j \mathbf{v}^\intercal, \mathbf{y}^{(j)} \rangle = \pm 1$.  Thus, for these bounding values of $\lambda$, an additional constraint will be active, implying that $\mathbf{U} + \lambda_j \mathbf{e}_i \mathbf{v}^\intercal$ will be on a lower dimensional face than $\mathcal{F}$.
\end{proof}

Having established Lemma \ref{lem:face}, we now show the following corollary which allows us to further characterize the optima of (\ref{eq:obj})--(\ref{eq:const}).
\begin{coro}
\label{cor:nonstrict}
Non-strict optima (which by Lemma \ref{lem:fp_boundary} must lie on a face $\mathcal{F}$) are restricted to an interval $\mathbf{I}$ as given in (\ref{eq:int}).  Further, the end points of this interval lie in a lower dimensional face.
\end{coro}
\begin{proof}
Suppose $\mathbf{U}$ is a non-strict optima in the interior of some face. By Lemma \ref{lem:face} and the multilinearity of the determinant, there must exist an interval $\mathbf{I}$ over which the objective function is constant valued.
\end{proof}

We also remark, informally, that Lemma \ref{lem:face} also implies that all strict optima will in fact be vertices of the feasible region.  By Lemma \ref{lem:face} and Corollary \ref{cor:nonstrict}, for any $\mathbf{U}$ that lies on the interior of a face, $\mathbf{U}$ is not a strict optimum.  Indeed, there must either be a direction $\mathbf{e}_i \mathbf{v}^\intercal$ along which $\mathbf{U}$ may move either to increase the value of the objective function or keep the objective function constant.   This result is not needed to complete Theorem \ref{thm:vertex}, but provides insight into the geometry of the problem.

Having characterized the set of optima in (\ref{eq:obj})--(\ref{eq:const}), we now turn our attention to characterizing the behavior of Algorithms \ref{alg:fit} and \ref{alg:fit_corner}. 
To do so, we will consider the gradient of the objective function, which is given in \cite{petersen2008matrix} as
\begin{equation}
\nabla \log | \det \mathbf{U} | = \left( \mathbf{U}^{-1} \right)^\intercal.
\end{equation}
We note that an alternative proof of Lemma \ref{lem:fp_boundary} follows from the
fact that the gradient is non-zero as long as $\mathbf{U}$ is finite.  

In order to understand how gradient descent will behave on the boundary of the
feasible region, we must consider directional derivatives for directions that
lie on the problem boundary.  Let $\Delta \in \mathbb{R}^{n \times n}$ be a
direction such that, for feasible $\mathbf{U}$, $\mathbf{U} + \Delta$ is
also feasible.  Gradient descent will terminate if, for all $\Delta$,
\begin{equation}
\left\langle \left(\mathbf{U}^{-1}\right)^\intercal, \Delta \right\rangle_F \leq 0,
\label{eq:direction}
\end{equation}
where $\langle \cdot, \cdot \rangle_F$ is the Frobenius inner product.  
Thus, (\ref{eq:direction}) is equal to
$\mathrm{vec}\left(\left(\mathbf{U}^{-1}\right)^\intercal\right)^\intercal
\mathrm{vec}(\Delta)$ or
$\mathrm{tr}( \left(\mathbf{U}^{-1}\right)^\intercal \Delta )$.

We now show that (\ref{eq:direction}) can only hold if each row has either $n-1$ 
or $n$ active, linearly independent constraints.  This lemma, as well as Corollary \ref{cor:nonstrict},
motivate Algorithm \ref{alg:fit_corner} as it implies that there may be
corner cases where Algorithm \ref{alg:fit} will fail, but these corner cases can
easily be handled by forcing the algorithm to terminate at the nearest vertex. 
\vspace{1mm}
\begin{lemma}
\label{lm:n_minus_one}
If any row of $\mathbf{U}$ has fewer then $n-1$ active constraints, 
then there exists a non-zero matrix $\mathbf{\Delta}$ such that 
$\mathbf{U + \Delta}$ satisfies (\ref{eq:const}) and $\log |\det \mathbf{U + \Delta}| > 
\log | \det \mathbf{U} |$.
\end{lemma}
\begin{proof}
Consider the $i$th row of $\mathbf{U}$, denoted $\mathbf{u}^{(i)}$, and let 
$\delta^{(i)}$ be the corresponding elements of $\Delta$.  The corresponding row of
the gradient matrix is given by:
\begin{equation}
\left( \nabla \mathbf{U} \right) ^{(i)} = \frac{1}{\det \mathbf{U}} \mathbf{c}^{(i)},
\end{equation}
where $\mathbf{c}^{(i)}$ is the $i$th row of the cofactor matrix of
$\mathbf{U}$.  Since $\mathbf{U}$ must be full rank, $\mathbf{c}^{(i)}$ must be
non-zero for all $i$.  Thus, for any $i$, the only way in which we could have
\begin{equation} 
\frac{1}{\det \mathbf{U}} \left\langle \mathbf{c}^{(i)}, \delta^{(i)}
\right\rangle_F = 0
\label{eq:dir_row}
\end{equation}
is if the only feasible values of $\delta^{(i)}$ (i.e. those that do not move
$\mathbf{U}$ outside the feasible region) are orthogonal to $\mathbf{c}^{(i)}$.

If fewer than $n-1$ linearly independent constraints are active on the $i$th
row, then there always is a subspace of at least dimension two from which we can select $\delta^{(i)}$.  
Precisely, suppose that the constraints given by $y_1, \ldots, y_{n-2}$ are active.  
The subspace spanned by these vectors must always have a null space of at least
dimension two; if $\delta^{(i)}$ is contained in this nullspace then $\mathbf{U + \Delta}$ will be
feasible. As long as this nullspace has dimension at least two, for all $i$, 
there exists a $\delta^{(i)}$ such that $\mathbf{U + \Delta}$ satisfies (\ref{eq:const}) 
and that has $\langle c^{(i)},\delta^{(i)} \rangle \neq 0$.  Thus, gradient descent will
always proceed as long as fewer than $n-1$ constraints are active on each row. 

\end{proof}

In other words, by Lemma \ref{lm:n_minus_one}, if gradient descent terminates and we are not at a vertex, there
must be at least one row with exactly $n-1$ active, linearly independent
constraints.  In this case, we may move along the interval $\mathbf{I}$ guaranteed by Corollary \ref{cor:nonstrict} until we reach a lower dimensional face.  
This lower dimensional face will either be a vertex, in which case we have 
reached a strict optima and the algorithm will terminate, or there will exist 
a positive gradient and we can resume gradient descent. This completes 
Theorem \ref{thm:vertex}.  

\section{Proof of Theorem \ref{thm:n3iff}}
\label{apx:n3}
\begin{customthm}{\ref{thm:n3iff}}
When $n=3$, Algorithm \ref{alg:fit_corner} is correct with probability 1 if and only if 
$k \geq 4$ and there exists a $3 \times 4$ matrix $\mathbf{V}$, such that $\text{cols}(\mathbf{V}) \subseteq \text{cols}(\mathbf{X})$, 
$\text{span}(\mathbf{v}_1,\ldots,\mathbf{v}_4) = \mathbb{R}^3$, and all vectors in $\mathbf{V}$ are pair-wise linearly independent.  
\end{customthm}
We know that if $\mathbf{X}$ has the maximal subset property, then Algorithm \ref{alg:fit_corner} will always terminate at a global maximum of (\ref{eq:obj})--(\ref{eq:const}) and that the set of global optima contains all solutions in the form $\mathbf{TX}$ for all $\mathbf{T} \in \mathcal{T}$.  However, for $n=3$, the maximal subset property alone does not ensure that all global optima will be solutions to the blind decoding problem.  Spurious optima must have the form $\mathbf{QX} \in [-1,+1]^{n \times k}$, where $\det \mathbf{Q} = \pm 1$ and $\mathbf{Q} \notin \mathcal{T}$.  Algorithm \ref{alg:fit_corner} will only be correct with probability 1 if there are no spurious optima.

We now show that, for $n=3$, if $\mathbf{X}$ has four distinct columns (\emph{distinct} meaning pair-wise linearly independent), $\mathbf{QX} \in [-1,+1]^{n \times k}$ and $\det \mathbf{Q} = \pm 1$ implies that $\mathbf{Q} \in \mathcal{T}$.  This further implies that all global optima are solutions to the blind decoding problem.
Consider the following choice of $\mathbf{X}$: 
\begin{equation}
\mathbf{X} = 
\begin{bmatrix}
1 & 1 & -1 & -1 \\
1 & -1 & 1 & -1 \\
1 & 1 & 1 & 1
\end{bmatrix}.
\label{eq:square}
\end{equation}
The following lemma shows that this choice of $\mathbf{X}$ further restricts $\mathbf{Q}$ to be orthogonal.
\begin{lemma}
\label{lem:orthogonal}
Suppose $\mathbf{Q}$ has $\det \mathbf{Q} = \pm 1$, and, for $\mathbf{X}$ given by (\ref{eq:square}), if $\mathbf{QX} \in [-1,+1]^{3 \times 3}$, then $\mathbf{Q} \in O(3)$.
\end{lemma}
\begin{proof}
We know that, by Theorem \ref{thm:vertex}, Algorithm \ref{alg:fit_corner} must terminate 
at a vertex.  Thus, we must have $\mathbf{QX} \in \{-1,+1\}^{n \times k}$, which further implies that, for all $i$, $\|\mathbf{Qx}_i\|_2 = \|\mathbf{x}_i\|_2$.  We now consider the set of linear operators with determinant $\pm 1$ whose action preserves the norms of each vector $\mathbf{x}_i$.

Define the operator $\tilde{\mathbf{Q}}=\mathbf{\Sigma V}^\intercal$, where $\mathbf{\Sigma}$ denotes the diagonal matrix containing the singular values of $\mathbf{Q}$ and $\mathbf{V}$ denotes the right singular vectors of $\mathbf{Q}$.  Since $\mathbf{Q} = \mathbf{U \tilde{Q}}$, for some $\mathbf{U} \in O(3)$, there must exist a $\tilde{\mathbf{Q}}$ such that $\tilde{\mathbf{Q}}\mathbf{x}_i = \pm \mathbf{x}_i$ for all $i$; otherwise, $\|\mathbf{Qx}_i\|_2 \neq \|\mathbf{x}_i\|_2$ for some $i$.

We can consider the action of any $\tilde{\mathbf{Q}}$ on the surface of a sphere of radius $\sqrt{3}$.  This sphere, $\mathcal{S}$, contains the vectors that comprise the columns of $\mathbf{X}$.
Under the action of any linear operator with determinant $\pm 1$, $\mathcal{S}$ will be mapped to an ellipsoid, $\mathcal{E}$, such that $\text{vol}(\mathcal{E}) = \text{vol}(\mathcal{S})$. Further, we know that $\mathcal{E}$, given by $\mathcal{\tilde{Q}}$, must contain the points $\mathbf{x}_i$ for all $i$.  Lemma \ref{lem:orthogonal} is now completed by Claim \ref{clm:sphere}, which shows that this is only possible if $\mathcal{E} = \mathcal{S}$, implying that $\mathbf{Q} \in O(3)$. 
\end{proof}
\begin{claim}
\label{clm:sphere}
The only ellipsoid that is centered on the origin, has volume $4\pi \sqrt{3}$, and contains the points given by $\text{cols}(\mathbf{X})$ is a sphere.
\end{claim}
\begin{proof}
Consider an arbitrary ellipsoid, $\mathcal{E}$, centered about the origin.  For some $\mathbf{A} \succeq 0$, this can be described by the following equation
\begin{align}
\mathbf{v}^\intercal \mathbf{A} \mathbf{v} &= 1, \nonumber \\
\begin{bmatrix}
x & y & z
\end{bmatrix}
\begin{bmatrix}
a & \frac{d}{2} & \frac{f}{2} \\
\frac{d}{2} & b & \frac{e}{2} \\
\frac{f}{2} & \frac{e}{2} & c  
\end{bmatrix}
\begin{bmatrix}
x \\
y \\
z
\end{bmatrix} &= 1. \label{eq:ellipse}
\end{align}
We require the ellipsoid to contain the points $(\pm 1, \pm 1, 1)$.  By substituting these points into (\ref{eq:ellipse}), it is seen that we must have $d = e = f =
0$ and $a+b+c=1$. Thus, $\mathbf{A}$ must be diagonal.

The eigenvalues of $\mathbf{A}$ are the inverse squares of the length of the
semi-axes of the ellipsoid.  This implies that the volume of the ellipsoid is given
by $\text{vol}(\mathcal{E}) = \frac{4}{3} \pi / \sqrt{abc}$.  The only solution
that gives $\text{vol}(\mathcal{E}) = 4 \pi \sqrt{3}$ with $a+b+c=1$ is $a = b = c =
\frac{1}{3}$, implying that the required ellipsoid in fact a sphere.
\end{proof}
Having established that (\ref{eq:square}) implies that $\mathbf{Q} \in O(3)$, we can further show that the only feasible elements of $O(3)$ are in fact the ATMs.
\begin{prop}
Suppose $\mathbf{Q}$ has $\mathbf{Q} \in O(3)$, and $\mathbf{QX} \in [-1,+1]^{3 \times 4}$, then $\mathbf{Q} \in \mathcal{T}$.
\end{prop}
\begin{proof}
Notice that the vectors $\text{cols}(\mathbf{X})$ form a face of the unit cube.  
Since $\mathbf{QX} \in [-1,+1]^{3 \times 4}$ and $\mathbf{Q} \in O(3)$, then $\text{cols}(\mathbf{QX})$ must also form the face of a unit cube.  This is because $\mathbf{Q}$ is orthogonal and must preserve norms and planes.  This fact restricts $\mathbf{Q}$ to the symmetries of the cube.  There are 48 symmetries of the cube, which correspond to the set of 48 ATMs.
\end{proof}
%

Finally notice that for all $\mathbf{D} = \text{diag}(\pm 1, \ldots, \pm 1)$, $\mathbf{QX} \in [-1,+1]^{n \times k}$ implies that $\mathbf{QXD} \in [-1,+1]^{n \times k}$.  Similarly, for all permutation matrices $\mathbf{P}$, $\mathbf{QXP}$ is feasible if and only if $\mathbf{QX}$ is feasible.  For $n=3$, all possible $\pm 1$-valued matrices that contain four distinct columns can be expressed as $\mathbf{XPD}$.  This implies that if $\mathbf{QX} \in [-1, +1]^{3 \times 4}$, for any possible $\mathbf{X}$ with four distinct columns, then $\mathbf{Q} \in \mathcal{T}$. Hence, any choice of $\mathbf{X}$ with four distinct columns is sufficient to ensure that there are no spurious optima.

We now turn our attention to showing the converse: that requiring $\mathbf{X}$ to have four distinct columns (that is, pairwise linearly independent) is in fact necessary for Algorithm \ref{alg:fit_corner} to be correct with probability 1.  First, if $n=k=3$, then for any choice of $\mathbf{X}$ such that $\mathbf{X}$ has the maximal subset property, spurious optima will exist.  For $n=3$, there is a small collection of three vectors, up to the ATMs, that have the maximal subset property, and so one may check that this is always true.  Therefore, when $n=k=3$, there will always be a matrix $\mathbf{Q}$ with unit determinant such that $\mathbf{QX} \in [-1,+1]^{3 \times 3}$ and $\mathbf{Q} \notin \mathcal{T}$.  

If $\mathbf{X}$ does not have four distinct columns then this will always be the case.
Consider the case where $k > 3$ and $\mathbf{X}$ has the maximal subset property but no four columns of $\mathbf{X}$ are distinct. 
Let the matrix $\mathbf{V}$ be formed from any subset of three columns of $\mathbf{X}$ such that $\mathbf{V}$ has the maximal subset property.  
Then we must have that for all $i$, there exists a $j$ 
such that $\mathbf{x}_i = \pm \mathbf{v}_j$. 
For any $\mathbf{Qv}_j \in [-1, +1]^n$, we also must have $-\mathbf{Qv}_j = \mathbf{Qx}_i \in [-1,+1]^n$.  Thus, whenever $\mathbf{X}$ does not contain any columns that are distinct from the columns of $\mathbf{V}$, then $\mathbf{X}$ will also have the same set of optima as $\mathbf{V}$.
This completes Theorem \ref{thm:n3iff}.
\section{Proof of Theorem \ref{thm:n4iff}}
\label{apx:theorem_three}
In this appendix we prove Theorem \ref{thm:n4iff}, which gives necessary and sufficient conditions so that Algorithms \ref{alg:fit} and \ref{alg:fit_corner} return correct solutions to the blind decoding problem when $n=4$.  Notice that because a Hadamard matrix exists at $n=4$, we know by Lemma \ref{cor:had_iff} the only optima are strict and are vertices of the feasible region.  Thus Algorithm \ref{alg:fit_corner} is not needed in this case.  
Before considering the specific case of $n=4$, we prove the following more
general statement for values of $n$ such that a Hadamard matrix exists. This
result will be used in the proof of Theorem \ref{thm:n4iff}. 
When a Hadamard matrix exists, we can further characterize the solutions to
(\ref{eq:obj})--(\ref{eq:const}) in the noiseless case as follows.
\begin{lemma}
\label{lm:ortho}
If $n$ is such that Hadamard matrix exists, and $\mathbf{X}$ has the maximal
subset property, then the only global optima to (\ref{eq:obj})--(\ref{eq:const})
on input $\mathbf{Y}=\mathbf{AX}$ are of the form $\mathbf{U}=\mathbf{QA}^{-1}$, 
where $\mathbf{Q} \in O(n)$.
\end{lemma}

The following claim is helpful in proving this lemma:
\begin{claim}
\label{clm:ortho}
For $\mathbf{X} \in \mathbb{R}^{n \times n}$ is an orthogonal matrix and some
$\mathbf{M} \in \mathbb{R}^{n \times n}$ with $| \det \mathbf{M} | = 1$, if
$\|\mathbf{Mx}_i\|_2 \leq \|\mathbf{x}_i\|_2$ for all $i$, then $\mathbf{M}$ must be 
orthogonal.
\end{claim}
\begin{proof}
Let $\mathbf{x}_1, \ldots, \mathbf{x}_n$ be the orthonormal basis obtained by
the columns of $\mathbf{X}$.  
\begin{align}
\mathrm{tr}\left( \mathbf{M}^T\mathbf{M} \right) &= \left\| \mathbf{M}^T \mathbf{M}
\right\|_F^2 = \sum_{i=1}^n |\lambda_i|^2 \\
\left| \sum_{i=1}^n \mathbf{x}_i^T \mathbf{M}^T \mathbf{M} \mathbf{x}_i \right| &\leq
\sum_{i=1}^n | \mathbf{x}_i^T \mathbf{M}^T \mathbf{M} \mathbf{x}_i | \\
	&= \sum_{i=1}^n |\mathbf{Mx}_i|^2 \\
	&\leq \sum_{i=1}^n |\mathbf{x}_i|^2 \\
    &= n
\end{align}
However, because $|\det \mathbf{M}| = 1$, this implies:
\begin{equation}
\det \mathbf{M}^T\mathbf{M} = \prod_{i=1}^n \lambda_i^2 = 1 = \sum_{i=1}^n
|\lambda_i|^2,
\end{equation}
and by the inequality of arithmetic and geometric means, this implies
$\lambda_i^2 = 1$ for all $i$. Since $\mathbf{M}$ is real valued, this implies
$\mathbf{M}$ is orthogonal.
\end{proof}

By Lemma \ref{lm:detconst}, we know that $\mathbf{Q}$ must have determinant one.
Let $\mathbf{V}$ be the matrix guaranteed by the maximal subset property.
$\mathbf{V}$ must be Hadamard.  The matrix $\mathbf{UAV}$, will also have the
maximum value of the determinant over all $[-1,+1]^{n \times n}$ matrices, and
thus must also be Hadamard.  Since all points in $\mathbf{V}$ and $\mathbf{UAV}$
have $\ell_2$-norm $2^{n/2}$, by Claim \ref{clm:ortho}, $\mathbf{UA}$ must be
orthogonal. This completes the proof of Lemma \ref{lm:ortho}.   

At this point, one might be tempted to conjecture that, in fact, the maximal subset 
property is on its own sufficient; that is that the orthogonal matrix $\mathbf{Q}$ in 
Lemma \ref{lm:ortho} can be replaced by an ATM $\mathbf{T}$. However, this is
not the case as we will show in the proof of Theorem \ref{thm:n4iff}, given below.

\begin{customthm}{\ref{thm:n4iff}}
When $n=4$, Algorithm \ref{alg:fit} is correct with probability 1 if and only if $k \geq 5$ and $\text{cols}(\mathbf{X})$
contains at least four linearly independent vectors from $\mathcal{H}_4^{(i)}$
and a fifth vector from $\mathcal{H}_4^{(j)}$ for any $i \neq j$. 

Algorithm \ref{alg:fit} will be correct with probability 0.5 if $\text{cols}(\mathbf{X})$ has only four
linearly independent vectors belonging to the same equivalence class.
\end{customthm}
We prove Theorem \ref{thm:n4iff} in two parts.  First, in Lemma
\ref{lm:n4_succ_prob}, we show that when $n=k$
and $\mathbf{X}$ has the maximal subset property, Algorithm \ref{alg:fit} will
return the correct solution to the blind decoding problem with probability 0.5,
and that with the addition of an extra vector from a separate equivalence class,
all global optima correspond to solutions to the blind decoding problem.
In Lemma \ref{lm:n4_succ_prob} we in fact prove a slightly more general statement and give
the probability of that all global optima correspond to solutions of the blind
decoding problem given the input to Algorithm \ref{alg:fit} is chosen
uniformly at random.  Second, in Lemma \ref{lm:n4_global}, we show that for all
values of $k$, all optima are indeed global despite the fact that there are
suboptimal vertices.
\begin{lemma}
\label{lm:n4_succ_prob}
For $n=4$, for a collection of $k$ samples chosen uniformly at random, the
probability that all global optima will correspond to solutions of the blind
decoding problem is given by
\begin{equation}
\mathrm{Pr}(\mathrm{Success}) = r(4,k) \cdot r(3,k)^4 \cdot (1 - 2^{n-k}).
\end{equation}
\end{lemma}
\begin{proof}
One can verify that for $n=4$, for a matrix to have the maximal subset
property (and hence be a Hadamard matrix), then not only must the matrix be full 
rank, but all matrices obtained by choosing a subset of three 
rows must also be full rank. The probability that a random set of $k$ vectors of
dimension 4 is a Hadamard matrix given by:
\begin{equation}
\mathrm{Pr}(\mathrm{Success}) = r(4,k) \cdot r(3,k)^4. \label{eq:n4}
\end{equation}
It can be seen that there is an orthogonal matrix, $\mathbf{S}$, that is not an ATM, such 
that for all $\mathbf{G} \in \mathcal{H}_4^{(1)}$, and $\mathbf{H} \in
\mathcal{H}_4^{(2)}$, $\mathbf{G}=\mathbf{SH}$.  To find an example of such a matrix, for any choice of $\mathbf{G}$ and $\mathbf{H}$, compute $\mathbf{S} = \mathbf{G}^{-1} \mathbf{H}$. This implies that, exactly half of the 
global optima are solutions to (\ref{eq:obj})--(\ref{eq:const}). This is
consistent with the observation that if $\mathbf{X}$ has a maximal subset, then Algorithm
\ref{alg:fit} succeeds 50\% of the time for $n=4, k=4$. 

Notice that for this same $\mathbf{G,H}$, and $\mathbf{S}$ given above, 
$\mathbf{SG} \notin \{-1,+1\}^4$ and similarly
$\mathbf{S}^{-1}\mathbf{H} \notin \{-1,+1\}^4$, for all $\mathbf{G}$ and
$\mathbf{H}$.  Notice further that all vectors in $\{-1,+1\}^4$ appear in either 
$\mathcal{H}_4^{(1)}$ or $\mathcal{H}_4^{(2)}$, and that the product of $\mathbf{S}$ 
times any vector in $\mathcal{H}_4^{(1)}$ is not in $\{-1,+1\}^4$. Similarly, any $\mathbf{S}^{-1}$ times any $\mathcal{H}_4^{(2)}$ is not in $\{-1,+1\}^4$.

Now consider a
collection of $k>n$ vectors that contains 4 independent elements of
$\mathcal{H}_4^{(i)}$, for some $i$.  If all vectors in this collection belong to the same
equivalence class, then matrices containing a factor of $\mathbf{S}$ will be the
optima of (\ref{eq:obj})--(\ref{eq:const}).  Otherwise, all such matrices will lie
outside of the feasible region and all global optima will correspond to solutions. 
For this reason, adding constraints removes vertices from the feasible region,
thus increasing the success probability of Algorithm 2.

Given the above argument, we have a probability of $1 - 2^{n-k}$ that the only global 
optima will be solutions, conditioned on the fact that the matrix has the maximal subset 
property.  Since equation (\ref{eq:n4}) gives the probability of the maximum
subset holding, we can express the probability
that, given $k$ random samples, all global optima are solutions to
(\ref{eq:obj})--(\ref{eq:const}) as:
\begin{equation}
\mathrm{Pr}(\mathrm{Success}) = r(4,k) \cdot r(3,k)^4 \cdot (1 - 2^{n-k}).
\end{equation}
\end{proof}
In order to arrive at our desired result for $n=4$, 
we still need to show that no vertices that correspond to matrices with determinant 
of $\pm 8$ are local optima of (\ref{eq:obj})--(\ref{eq:const}).  This is proven below and 
completes Theorem \ref{thm:n4iff}.

\begin{lemma}
\label{lm:n4_global}
For $n=4$, all optima of (\ref{eq:obj})--(\ref{eq:const}) are global.
\end{lemma}
\begin{proof}
Optima of (\ref{eq:obj})--(\ref{eq:const}) can only lie on vertices of the problem boundary --- that is
optima can only correspond to non-singular $\{\pm 1\}$-valued matrices by Lemma
\ref{cor:had_iff}. We need
to show that no matrices with determinant $\pm 8$ are local optima. 
We begin with two facts which have been verified through computer simulation.

Create a graph with a node for each matrix in $\{-1,1\}^{4 \times 4}$ and edges
between each pair of nodes that have Hamming distance of one.
Remove from this graph all nodes which correspond two determinant zero,
leaving only nodes with determinant $\pm 8$ and $\pm 16$. Since the determinant
is a linear function of the columns of a matrix, then the value of $\det
\mathbf{A}$
changes linearly along each edge of this graph.  Studying the geometry of this graph will give us insight into paths that Algorithm \ref{alg:fit_corner} may travel in arriving to an optima.

The first observation is that the graph is partitioned into two components. This
can be verified, for example, either by inspecting the least eigenvalues of the Laplacian of the
adjacency matrix of the graph or performing a breadth-first search.  We find that one component of this graph corresponds to
$\{\pm 1\}$-valued matrix that have positive determinants, and the other to all 
all matrices with negative determinants.  This means that one can traverse either component of these graphs without the objective function changing sign.  This further implies that, over each edge of the graph, 
the objective function is either constant (and equal to $\log 8$ along each edge), or
changes logrithmically from $\log 8$ to $\log 16$.

It can further be verified that the maximum Hamming distance between any matrix
of determinant $\pm 8$ and determinant $\pm 16$ is 2.  Those that have Hamming
distance 1 are clearly not local optima.  Thus we turn our attention to the
remaining determinate $\pm 8$ matrices that have Hamming distance 2. Each of
matrices is connected to at least two determinant $\pm 8$ matrices that are
distance 1 away from a $\pm 16$ matrix.  This is depicted in Figure
\ref{fig:vertices}: matrix $A$ is distance 2 from optimal matrix $D$ and
adjacent to suboptimal matrices $B$ and $C$. 

From Corollary \ref{cor:nonstrict}, we know that the objective function can
only be constant on an affine subspace.  It is constant along the lines $AB$ and
$AC$, but not on the line between $BD$ and $CD$.  This implies that the
objective function cannot be constant on the line $AD$.  Further, we know that the
determinants of $A$ and $D$ have the same sign and that there cannot be critical
points on this line. Thus the objective function along this line is
monotonically increasing, implying that no determinant $\pm 8$ matrix is a local
optimum.
\end{proof}
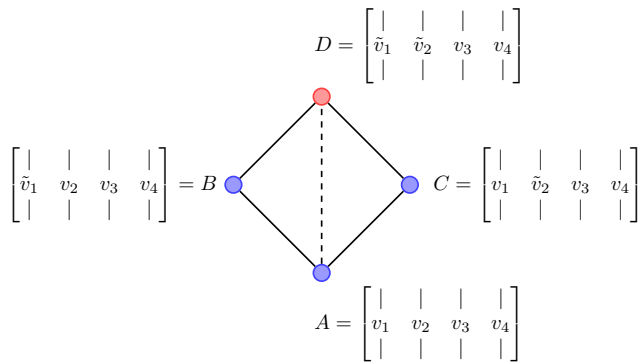
\begin{figure} 
    \begin{centering}

\resizebox{\columnwidth}{!}{%
\begin{tikzpicture}[
red_point/.style={circle,draw=red!75,fill=red!40,thick,
inner sep=0pt,minimum size=8pt},
blue_point/.style={circle,draw=blue!75,fill=blue!40,thick,
inner sep=0pt,minimum size=8pt}
]
	\node[blue_point,label=below right:
		{   \hspace{-6mm} 
			$A=\begin{bmatrix}
				\vert & \vert & \vert & \vert \\ 
				v_1 & v_2 & v_3 & v_4 \\
				\vert & \vert & \vert & \vert 
				\end{bmatrix}$
		}] (a) at (0,0) {};

	\node[blue_point,label=left:
		{
			$\begin{bmatrix}
			\vert & \vert & \vert & \vert \\
			\tilde{v}_1 & v_2 & v_3 & v_4 \\
			\vert & \vert & \vert & \vert
			\end{bmatrix}=B$
		}] (b) at (-1.5,1.5) {};

	\node[blue_point,label=right:
		{			
			$C = \begin{bmatrix}
			\vert & \vert & \vert & \vert \\
			v_1 & \tilde{v}_2 & v_3 & v_4 \\
			\vert & \vert & \vert & \vert
			\end{bmatrix}$ 
		}] (c) at (1.5,1.5) {};

	\node[red_point,label=above right:
		{
			\hspace{-6mm}
			$D=\begin{bmatrix}
			\vert & \vert & \vert & \vert \\
			\tilde{v}_1 & \tilde{v}_2 & v_3 & v_4 \\
			\vert & \vert & \vert & \vert
			\end{bmatrix}$
		}] (d) at (0,3) {};

	\draw[thick] (a) -- (b);
	\draw[thick] (b) -- (d);
	\draw[thick] (a) -- (c);
	\draw[thick] (c) -- (d);
	\draw[thick,dashed] (a) -- (d);

\end{tikzpicture}
}%
    \par\end{centering}
    \caption{The matrix $A$ with determinant 8 is Hamming distance 2 away from 
the matrix $D$ with determinant 16.  $A$ is Hamming distance 1 away from both
$B$ and $C$ which also have determinant 8.  The objective function is constant
on the lines $AB$ and $AB$ but monotonically increasing on the line $AD$. }
    \label{fig:vertices}
\end{figure}

\section{Distribution of the Rank of a Collection of Random Vectors}
\label{sec:rank}
In this section, we consider the distribution of the rank of a collection of
vectors drawn uniformly from $\mathbb{F}_q^n$.  Because the rank of a collection
of vectors in $\mathbb{F}_q^n$ is less than or equal to its rank in
$\mathbb{R}^n$, this will allow us to obtain an exact expression for $r(n,k)$,
which is the probability that a set of binary vectors is full rank over
$\mathbb{R}^n$.  We begin by stating a simple lower bound which shows that the
probability of a collection of vectors not being full rank decays exponentially
fast as the number of vectors grow. In this section we refer to $k$ as the size
of the collection of vectors.

Noting that the number of subspaces of dimension $n-1$ is $2^n$, and that the
probability that all $k$ vectors live in any single $n-1$ dimensional subspace
is $(1/2)^k$, by union bound, we have the following probability:
\begin{align}
\Pr(\mathrm{not\,full\,rank\,in\,}\mathbb{R}^n) &\leq
\Pr(\mathrm{not\,full\,rank\,in\,
}\mathbb{F}_2^n) \\
 & \leq 2^n \left( \frac{1}{2} \right)^k = 2^{n-k}
\end{align}

We now compute the exact distribution of the dimension of the subspace spanned by
a random subset of a vector space over a finite field. From this distribution, we
can compute an exact expression for $r(n,k)$.  The computation makes use of the
M\"obius inversion formula, a standard tool in combinatorics and number theory
that provides a natural way to count elements of partially ordered sets using an
overcounting-undercounting procedure.  For a full treatment on M\"obius
functions and their applications see \cite{BG75}.

In \cite{BG75}, the authors apply M\"obius inversion to counting vector
subspaces.  If $\mathbf{U}$ and $\mathbf{V}$ are subspaces of $\mathbb{F}_q^n$,
then $\mathbf{U} \leq \mathbf{V}$ iff $\mathbf{U}$ is a subspace of
$\mathbf{V}$. This relation forms a partial ordering for all subspaces of
$\mathbf{F}_q^n$.    

We are interesting in counting the number of collections of vectors which span a
given subspace.  
Let $N_=(\mathbf{X})$ be the number of $k$-tuples ($\mathbf{x}_1, \ldots, \mathbf{x}_k$) that 
span the subspace $\mathbf{X}$, and let $N_\leq(\mathbf{X})$ be the number of $k$-tuples that 
span either $\mathbf{X}$ or a subspace of $\mathbf{X}$.  Clearly,
\begin{equation}
N_\leq(\mathbf{X}) = \sum_{\mathbf{U}:\mathbf{U} \leq \mathbf{X}} N_=(\mathbf{U})
\end{equation} 
Note that the function $N_\leq(\mathbf{U})$ is easily computable as:  
\begin{equation}
N_\leq(\mathbf{U}) = q^{(\dim \mathbf{U})k}.
\end{equation}
The M\"{o}bius inversion formula gives us a way to compute $N_=(\mathbf{U})$
through $N_\leq(\mathbf{U})$, namely:
\begin{equation}
N_=(\mathbf{X}) = 
\sum_{\mathbf{U}: \mathbf{U} \leq \mathbf{X}} 
\mu (\mathbf{U}, \mathbf{X}) N_\leq(\mathbf{U}),
\end{equation}
where $\mu$ is M\"obius function, which is the integer-valued function on ordered
pairs of subspaces defined implicitly by:
\begin{equation}
\sum_{\mathbf{X}: \mathbf{U} \leq \mathbf{X} \mathbf{W}} 
\mu(\mathbf{U},\mathbf{X}) =
\begin{cases}
1,\,\mathrm{if}\,\mathbf{U} = \mathbf{W} \\
0,\,\mathrm{if}\,\mathbf{U} \neq \mathbf{W}.
\end{cases}
\end{equation}
$\mu$ may be computed recursively by the following formula:
\begin{equation}
\mu (\mathbf{U},\mathbf{X}) = 
\begin{cases}
1,\,\mathbf{U} = \mathbf{W} \\
- \sum_{\mathbf{X}: \mathbf{U} \leq \mathbf{X} \leq \mathbf{W}} \mu (\mathbf{U},
\mathbf{W}),\,\mathbf{U} < \mathbf{W} \\
0,\,\mathrm{otherwise}
\end{cases}
\end{equation}
In \cite{BG75}, the authors show that $\mu(\mathbf{U}, \mathbf{X})$ depends only
on the difference between $\dim \mathbf{X}$ and $\dim \mathbf{U}$.  Letting $i =
\dim \mathbf{X} - \dim \mathbf{U}$, they further show that:
\begin{equation}
\mu(\mathbf{U}, \mathbf{X}) = \mu_i = (-1)^i q^{i \choose 2}.
\end{equation}
With these preliminaries, we now have everything we need to prove the following
theorem:
\begin{theorem}
Let the set of vectors $\mathbf{x}_1, \ldots, \mathbf{x}_k$ be chosen uniformly at
random from $\mathbb{F}_q^n$, and let $\mathbf{W}$ denote the subspace spanned by 
these vectors. Then the probability that $\dim \mathbf{W} = m$ is:
\begin{equation}
\mathrm{Pr}(\dim \mathbf{W} = m) = q^{-nk} {n \brack m}_q \sum_{i=0}^m (-1)^i q^{i \choose 2}
q^{(m-i)k} {m \brack i}_q. \label{eq:subspace}
\end{equation}
\end{theorem}

\begin{proof} 
Note that there are $q^{nk}$ total possible 
$k$-tuples of ($\mathbf{x}_1, \ldots, \mathbf{x}_k$). This implies that 
\begin{equation}
\mathrm{Pr}(\dim \mathbf{X} = m) = \sum_{\mathbf{X} \leq \mathbb{F}_q^n: \dim
\mathbf{X} = m}
\frac{N_=(\mathbf{X})}{q^{nk}}.
\end{equation}
We now compute $N_=(\mathbf{X})$ through the use of M\"{o}bius inversion, using the
expressions for $N_=(\mathbf{X})$, $N_\leq(\mathbf{X})$ and $\mu_i$ we derived above, we have:
\begin{align}
&\mathrm{Pr}(\dim \mathbf{W} = m) 
	= \sum_{\mathbf{X} \leq \mathbb{F}_q^n: \dim \mathbf{X} = m} 
	   \frac{N_=(\mathbf{X})}{q^{nk}} \\
	&= \frac{1}{q^{nk}} \sum_{\mathbf{X} \leq \mathbb{F}_q^n: \dim \mathbf{X} = m} 
	   \: \sum_{\mathbf{U}:\mathbf{U} \leq \mathbf{X}} 
		\mu_{\dim \mathbf{X} - \dim \mathbf{U}} N_\leq (\mathbf{U}) \\
 	&= \frac{1}{q^{nk}} \sum_{\mathbf{X} \leq \mathbb{F}_q^n: \dim \mathbf{X} = m} 
	    \: \sum_{i=0}^m 
	    \: \sum_{\mathbf{U} \leq \mathbf{X}, \dim \mathbf{U} = m-i} (-1)^i 
		q^{i \choose 2} q^{(m-i)k} 
		\label{eq:bracket1} \\
    &= q^{-nk} {n \brack m}_q \, \sum_{i=0}^m (-1)^i q^{i \choose 2}
		 q^{(m-i)k} {m \brack i}_q \label{eq:bracket2}
\end{align}
where (\ref{eq:bracket1}) and (\ref{eq:bracket2}) follow by noting that ${n
\brack m}_q$ counts the number of $m$-dimensional subspaces of $\mathbb{F}_q^n$.
\end{proof}
From this theorem, an expression for $r(n,k)$ readily follows by substituting 
$m=n$ and $q=2$:
\begin{equation}
r(n,k) = 2^{-nk} \sum_{i=0}^n (-1)^i 2^{i \choose 2} 2^{(n-i)k} {n \brack i}_2. 
\label{eq:nkrank}
\end{equation}
\section*{Acknowledgements}
The authors would like to thank Yonathan Morin for insightful conversation on 
MIMO decoding and channel estimation, and for his comments on a preliminary version
of this work,  Mainak Chowdhury for discussion on non-coherent MIMO channels and
optimization, Milind Rao for discussion on optimization, blind-source separation
and statistical learning techniques, Ronny Hadani for his comments on the use of our
algorithm in complex-valued channels, and Jonathan Perlstein for providing a 
counterexample for the $n=3, k=3$ case, and for his comments on a preliminary version of 
this work.

\balance
\bibliographystyle{ieeetr}
\bibliography{IEEEabrv,trdean}

\end{document}